\documentclass[12pt,onecolumn,draftclsnofoot]{IEEEtran}%%option setup
\IEEEoverridecommandlockouts
\usepackage{color}

\usepackage{times}
\usepackage[final]{graphicx}
\usepackage{amsmath,amsfonts}
\usepackage{amsthm}
\usepackage{amssymb,amsbsy}

\usepackage{cite}
\usepackage{multirow}
\usepackage{subcaption}

\newcommand{\ignore}[1]{}
\newcommand{\abs}[1]{\left\vert#1\right\vert}

\DeclareMathOperator{\bff}{\bf f}
\DeclareMathOperator{\bg}{\bf g}

\DeclareMathOperator{\bs}{\bf s}

\DeclareMathOperator{\bu}{\bf u}
\DeclareMathOperator{\bv}{\bf v}

\DeclareMathOperator{\bF}{\bf F}

\DeclareMathOperator{\bH}{\bf H}
\DeclareMathOperator{\bI}{\bf I}

\DeclareMathOperator{\trace}{\mathsf{Tr}}

\DeclareMathOperator{\complexs}{\mathbb{C}}
\DeclareMathOperator{\expec}{\mathcal{E}}
\newcommand{\ind}[1]{{(#1)}}

\newcommand{\hspp}{\hspace{0.05in} }
\newcommand{\hsppp}{\hspace{0.02in} }

\newtheorem{lem}{Lemma}

\newtheorem{prop}{Proposition}
\newtheorem{thm}{Theorem}

\newsavebox{\savepar}

\begin{document}
%\title{Channel Reconstruction for Beamforming in Millimeter Wave Multi-User MIMO Systems}
\title{Channel Reconstruction-Based Hybrid Precoding for Millimeter Wave Multi-User MIMO
Systems}
\author{\large Miguel R.\ Castellanos$^{\dagger}$, Vasanthan Raghavan$^{\star}$,
Jung H.\ Ryu$^{\star}$, \\ Ozge H. Koymen$^{\star}$, Junyi Li$^{\star}$, David J.\ Love$^{\dagger}$, and Borja Peleato$^{\dagger}$\\
$^{\dagger}$Purdue University, West Lafayette, IN 47907, USA \\
$^{\star}$Qualcomm Corporate R\&D, Bridgewater, NJ 08807, USA
%E-mail: {\tt \{castellm, djlove\}@purdue.edu}, {\tt vasanthan\_raghavan@ieee.org},
%\\
%{\tt \{vraghava, jungr, okoymen, junyil\}@qualcomm.com}
%\\ E-mail: {\tt castellm@purdue.edu, vasanthan\_raghavan@ieee.org}
\thanks{This material is based upon work supported in part by the National Science
Foundation under grants CCF1403458 and CNS1642982.}
}

\maketitle
%\IEEEpeerreviewmaketitle
%\vspace{-50mm}
%\baselineskip 17pt

\begin{abstract}
\noindent
The focus of this paper is on multi-user multi-input multi-output (MIMO)
transmissions for millimeter wave systems with a hybrid precoding architecture
at the base-station. To enable multi-user transmissions, the base-station uses
a cell-specific codebook of beamforming vectors over an initial {\em beam
alignment} phase. Each user uses a user-specific codebook of beamforming vectors
to learn the top-$P$ (where $P \geq 1$) beam pairs in terms of the observed
signal-to-noise ratio (${\sf SNR}$) in a single-user setting. The top-$P$ beam
indices along with their ${\sf SNR}$s are fed back from each user and the base-station
leverages this information to generate beam weights for simultaneous transmissions.
A typical method to generate the beam weights is to use {\em only} the best beam for
each user and either steer energy along this beam, or to utilize this information to
reduce multi-user interference. The other beams are used as fall back options to
address blockage or mobility. Such an approach completely discards information learned
about the channel condition(s) even though each user feeds back this information. With
this background, this work develops an advanced {\em directional} precoding structure for
simultaneous transmissions at the cost of an additional marginal feedback overhead. This
construction relies on three main innovations: 1) Additional feedback to allow the
base-station to {\em reconstruct} a rank-$P$ approximation of the channel matrix between
it and each user, 2) A zeroforcing structure that leverages this information to combat
multi-user interference by remaining agnostic of the receiver beam knowledge in the 
precoder design, and 3) A hybrid precoding architecture that allows {\em both}
amplitude and phase control at low-complexity and cost to allow the implementation of the
zeroforcing structure. Numerical studies show that the proposed scheme results in
a significant sum rate performance improvement over na{\"i}ve schemes even with a coarse
initial beam alignment codebook.
\end{abstract}


\section{Introduction}
\label{sec1}
Over the last few years, there has been a growing interest in leveraging the opening up
of the spectrum in the millimeter wave band ($\sim {\hspace{-0.03in}} 30$-$100$ GHz) in
realizing the emerging higher data rate demands of cellular
systems~\cite{khan,qualcomm,rappaport,boccardi1}. Communications in the millimeter wave
band suffers from increased path loss exponents, higher shadow fading, blockage and
penetration losses, etc., than sub-$6$ GHz systems leading to a poorer link margin than
legacy
systems~\cite{5G_whitepaper,3gpp_CM_rel14_38901,azar,shu_sun2016,vasanth_tap2018,vasanth_tap_blockage}.
However, by restricting attention to small cell coverage and by reaping the increased
array gains from the use of large antenna arrays at both the base-station and user ends,
significant rate improvements can be realized in practice.

%In addition to the use of large antenna arrays, m
Millimeter wave propagation is spatially {\em sparse} with few dominant clusters in
the channel relative to the number of
%antennas~\cite{hao_xu,5G_whitepaper,3gpp_CM_rel14_38901,samimi_conf,vasanth_it2,qualcomm2}.
antennas~\cite{5G_whitepaper,3gpp_CM_rel14_38901,vasanth_it2,qualcomm2}.
Spatial sparsity of the channel along with the use of large antenna arrays motivates
a subset of physical layer beamforming schemes based on
{\em directional} transmissions for signaling. In this context, there have been a
number of studies on the design and performance analysis of directional
beamforming/precoding structures for single-user multi-input multi-output (MIMO) systems~\cite{rusek,hur,roh,brady_tcom,oelayach,vasanth_gcom15,raghavan_jstsp,rangan,ghosh,sun}.
These works~\cite{brady_tcom,oelayach,vasanth_gcom15,raghavan_jstsp} show that
directional schemes are not only good from an implementation standpoint, but are
also robust to phase changes across clusters and allow a smooth tradeoff between
peak beamforming gain and initial user discovery latency. There has also been
progress in generalizing such directional constructions for multi-user MIMO transmissions~\cite{sun,vasanth_gcom16,vasanth_jsac2017,ang_li}.

In this context, while legacy systems use as many radio frequency (RF)
chains\footnote{An RF chain includes (but is not limited to) analog-to-digital
converters (ADCs), digital-to-analog converters (DACs), mixers, low-noise and
power amplifiers (PAs), etc.} as the number of antennas, their higher cost,
energy consumption, area and weight at millimeter wave carrier frequencies has
resulted in the popularity of {\em hybrid beamforming}
systems~\cite{molisch1,molisch2,venkateswaran,adhikary}. A hybrid beamforming
system uses a smaller number of RF chains than the number of antennas, with the
one extreme case of a single RF chain being called the analog/RF beamforming
system and the other extreme of as many RF chains as the number of antennas
being called the digital beamforming system. Spatial sparsity of millimeter
wave channels ensures that having as many RF chains as the number of dominant
clusters in the channel is sufficient to reap the full array gain possible over
these channels.

A number of recent works have addressed hybrid beamforming for millimeter wave systems.
The problem of finding the optimal precoder and combiner with a hybrid architecture
is posed as a sparse reconstruction problem in~\cite{oelayach}, leading to algorithms
and solutions based on basis pursuit methods. While the solutions achieve good performance
in certain cases, to address the performance gap between the solution proposed
in~\cite{oelayach} and the unconstrained beamformer structure, an iterative scheme is
proposed in~\cite{alkhateeb,alkhateeb2} relying on a hierarchical training codebook for
adaptive estimation of millimeter wave channels. The authors in~\cite{alkhateeb,alkhateeb2}
show that a few iterations of the scheme are sufficient to achieve near-optimal performance.
In~\cite{sohrabi}, it is established that a hybrid architecture can approach the
performance of a digital architecture as long as the number of RF chains is twice that
of the data-streams. A heuristic algorithm with good performance is developed when this
condition is not satisfied. A number of other works such
as~\cite{noh16fdbk,bogale,gao_hybrid,magueta_hybrid} have also explored iterative/algorithmic
solutions for hybrid beamforming.

A common theme that underlies most of these works is the assumption of {\em phase-only}
control in the RF/analog domain for the hybrid beamforming architecture. This
assumption makes sense at the user end with a smaller number of antennas (relative to the
base-station end), where operating the PAs below their peak rating across RF chains can lead to
a substantially poor uplink performance. On the other hand, amplitude control (denoted as
{\em amplitude tapering} in the antenna theory literature) is necessary at the base-station end
with a large number of antennas for side-lobe management and mitigating out-of-band emissions.
Further, given that the base-station is a network resource, simultaneous amplitude and phase
control of the individual antennas across RF chains is feasible at millimeter wave base-stations
at a low-complexity\footnote{Any calibration complexity can be seen as a one-time effort 
at the unit level for a large array and defrayed as a low network cost.} and 
cost~\cite[pp.\ 285-289]{harish},~\cite{huang2015low,briqech2017low}.
In particular, the millimeter
wave experimental prototype demonstrated in~\cite{vasanth_comm_mag_16} allows simultaneous
amplitude and phase control. Thus, it is important to consider a hybrid architecture with
these constraints. Further, given the directional nature of the channel, a solution
should both inherit a directional structure and provide an intuitive description of
the beam weights. For example, a {\em black box}-type algorithmic solution that does
not provide an intuitive description of the beam weights is less preferable over a
solution that is constructed out of measurement reports obtained over an initial
{\em beam alignment} phase with a directional structure for the sounding beams.

\noindent {\bf \em \underline{Main Contributions:}}
With this backdrop, this work addresses these two fundamental issues in hybrid beamformer design.
It is assumed that the base-station trains all the users in the cell with a cell-specific
codebook of beamforming vectors over an initial beam alignment phase. Each user makes an
estimate\footnote{In a practical implementation such as the Third Generation Partnership
Project New Radio (3GPP 5G-NR) design, $P = 4$ is typically assumed both in terms of
measurements and reporting~\cite{3gpp_CM_rel14_38912}. The received ${\sf SNR}$ is estimated
as the received power of a beamformed link (corresponding to the beam pair under consideration)
using a certain reference symbol resource. This metric is typically known as the reference
symbol received power (RSRP) of the link.} of the top-$P$ (where $P \geq 1$) beams over this
phase and reports the beam indices to be used by the base-station as well as the measured/received
signal-to-noise ratios (${\sf SNR}$s). The simplest implementation at the base-station uses
only the best beam information for beam steering or zeroforcing as
in~\cite{vasanth_gcom16,vasanth_jsac2017}, with other beams serving as fall back options.

In contrast to this approach, we propose to {\em reconstruct} or {\em estimate} a rank-$P$
approximation of the channel matrix between the base-station and the user (at the base-station
end). To realize
this reconstruction, we envision the additional feedback of the phase of the received signal
estimate of the top-$P$ beams over the beam alignment phase and the cross-correlation information
of the top-$P$ beams at the user end with the beam used for multi-user reception. With 
this novel construction, the base-station can remain agnostic of the user's top-$P$ 
beams in precoder design. In terms of overhead,
in 3GPP 5G-NR, these quantities can be fed back over the physical uplink control channel (PUCCH)
with a Type-II feedback scheme~\cite[Sec.\ 8.2.1.6.3, pp.\ 24-26]{3gpp_CM_rel14_38912}; see
Sec.~\ref{sec5b} for a detailed study that demonstrates this feedback overhead to be marginal.
Leveraging the rank-$P$ channel approximation, we propose the use of a zeroforcing structure
that is then quantized to meet the RF precoding constraints (amplitude and phase control) at
the base-station end for simultaneous transmissions.

To benchmark and compare the performance of the proposed scheme, we establish two upper bounds
for the sum rate. This is a fundamentally difficult problem given the non-convex dependence of
the sum rate on the beamforming vectors~\cite{cioffi_weighted,kobayashi,vasanth_arxiv2011}.
The first bound is based on an intuitive parsing and understanding of the zeroforcing
structure. The second bound is based on an alternating optimization of the
beamformer-combiner pair with signal-to-leakage and noise ratio
(${\sf SLNR}$)~\cite{tarighat} and signal-to-interference and noise ratio
(${\sf SINR}$) as optimization metrics. Numerical studies show that the proposed
scheme performs significantly better than a na{\"i}ve beam steering solution even for an initial
beam alignment codebook of poor resolution. Further, the proposed scheme is comparable with the
established upper bounds provided the beam alignment codebook resolution is moderate-to-good.
Thus, our work establishes the utility and efficacy of the proposed feedback techniques
as well as opens up avenues for further investigation of such approaches in
hybrid beamforming with millimeter wave systems.

\noindent {\bf \em \underline{Organization:}} This paper is organized as follows.
Sec.~\ref{sec2} develops the system setup and explains the RF precoder architectural
constraints adopted in this work. In Sec.~\ref{sec3}, we provide a background of the
initial beam alignment phase and the feedback mechanism necessary for the multi-user
beamforming envisioned in this work. Sec.~\ref{sec4} generates two upper bounds on
the sum rate to benchmark the performance of the proposed scheme. Sec.~\ref{sec5}
performs a number of numerical studies to understand the performance of the proposed
scheme relative to a na{\"i}ve beam steering solution as well as to the upper bounds
developed in Sec.~\ref{sec4}. Concluding remarks are provided in Sec.~\ref{sec6}.

\ignore{
\noindent {\bf \em \underline{Notations:}} Lower- and upper-case bold symbols are used to
denote vectors and matrices, respectively. The $i$-th entry of a vector ${\bf x}$ and the
$(i,j)$-th entry of a matrix ${\bf X}$ are denoted by ${\bf x}(i)$ and ${\bf X}(i,j)$,
respectively. The regular matrix transpose and complex conjugate Hermitian transpose
operations of a matrix are denoted by $(\cdot)^{\sf T}$ and $( \cdot)^{\dagger}$, respectively.
The inverse and principal square root of a square and a positive semi-definite matrix are denoted
as $(\cdot)^{-1}$ and $(\cdot)^{1/2}$, respectively. The trace and expectation operators
are denoted by ${\sf Tr}(\cdot)$ and $\expec[\cdot]$, respectively. The two-norm of a vector is
denoted as $\| \cdot \|$ with ${\mathbb{C}}$, ${\mathbb{R}}$ and ${\cal CN}$ standing for the set
of reals, complex and the complex normal random variable, respectively. Quantization operation of
an underlying variable with $B$ bits will commonly feature in this work and this operation is
denoted as ${\cal Q}_B(\cdot)$.
}

\noindent {\bf \em \underline{Notations:}} Lower- and upper-case bold symbols are used to
denote vectors and matrices, respectively. The $i$-th entry of a vector ${\bf x}$ and the
$(i,j)$-th entry of a matrix ${\bf X}$ are denoted by ${\bf x}(i)$ and ${\bf X}(i,j)$,
respectively. The regular matrix transpose and complex conjugate Hermitian transpose
operations of a matrix are denoted by $(\cdot)^{\sf T}$ and $( \cdot)^{\dagger}$, respectively.
The two-norm of a vector is denoted as $\| \cdot \|$ with ${\mathbb{C}}$, ${\mathbb{R}}$ and
${\cal CN}$ standing for the set of reals, complex numbers and the complex normal
random variable, respectively.

\section{System Setup}
\label{sec2}
We consider a cellular downlink scenario with a single base-station serving
$K_{\sf cell}$ potential users. The base-station and each user are assumed to
be equipped with planar arrays of dimensions $N_{\sf tx} \times N_{\sf tz}$
antennas and $N_{\sf rx} \times N_{\sf rz}$ antennas, respectively. At both
ends, the inter-antenna element spacing is $\lambda/2$ where $\lambda$ is the
wavelength of propagation. With $N_{\sf t} = N_{\sf tx} \cdot N_{\sf tz}$ and
$N_{\sf r} = N_{\sf rx} \cdot N_{\sf rz}$, the base-station and each user are
assumed to have $M_{\sf t} \leq N_{\sf t}$ and $M_{\sf r} \leq N_{\sf r}$ RF
chains, respectively.
%Fig.~\ref{fig_system_model} illustrates the system setup
%in the simple case of $K_{\sf cell} = 2$.

\ignore{
\begin{figure*}[htb!]
\begin{center}
\includegraphics[height=3.2in,width=6.3in]  {figures/systemFigs.eps}
\caption{\label{fig_system_model}
System model capturing a cellular downlink scenario where a base-station
communicates with $K_{\sf cell} = 2$ users over an initial beam alignment
phase (beam patterns in dashed green), followed by multi-user beam design
(illustrated in red) based on the feedback from each user.}
\end{center}
\end{figure*}
}

%We focus on the narrowband aspects in this work.
For the channel ${\bf H}_k \in \complexs^{N_{\sf r} \times N_{\sf t}}$ between
the base-station and the $k$-th user (where $k = 1, \cdots, K_{\sf cell}$), we
assume an extended geometric propagation model over $L_k$
%clusters/paths~\cite{saleh,3gpp_CM_rel14_38901,metis2020_tcom}
clusters/paths~\cite{saleh,3gpp_CM_rel14_38901}
\begin{eqnarray}
\label{eq:channel}
\bH_k = \sqrt{\frac{N_{\sf r} N_{\sf t} }{L_k}} \sum_{\ell = 1}^{L_k}
\alpha_{k,\ell} \bu_{k,\ell} \bv_{k,\ell}^{\dagger}.
\end{eqnarray}
In~(\ref{eq:channel}), $\alpha_{k,\ell}$, $\bu_{k,\ell}$ and $\bv_{k,\ell}$ denote
the complex gain, the array steering vector at the user end corresponding to
the angle of arrival (AoA) in azimuth/zenith, and the array steering vector at the
base-station corresponding to the angle of departure (AoD) in azimuth/zenith,
respectively. The
cluster gains are assumed to be independent and identically distributed (i.i.d.)
standard complex Gaussian random variables: $\alpha_{k,\ell} \sim \mathcal{CN}(0,1)$.
The normalization of the channel ensures that $\expec \big[
\trace(\bH_k \bH_k^\dagger ) \big] = N_{\sf r} N_{\sf t}$.

In terms of the system model, we focus on the narrowband aspects and assume that the
base-station serves $K \leq K_{\sf cell}$ users simultaneously with data along $M_{\sf t}$ RF
chains. The base-station precodes $r_m$ data-streams for the $m$-th user with
the $r_m \times 1$ symbol vector %$\bs_k \in \complexs^{r_k}$
${\bf s}_m$ using the $M_{\sf t} \times r_m$ digital/baseband precoder
${\bf F}_{ {\sf Dig},\hsppp m}$ which is then up-converted to the carrier
frequency by the use of the $N_{\sf t} \times M_{\sf t}$ RF precoder
${\bf F}_{\sf RF}$. This results in the following system equation at the $k$-th user
\begin{eqnarray}
{\bf y}_k = \sqrt{\frac{\rho}{K}} \hsppp {\bf H}_k {\bf F}_{\sf RF} \cdot
\left[ \sum_{m = 1}^K {\bf F}_{ {\sf Dig},\hsppp m} {\bf s}_m \right] + {\bf n}_k
\label{eq_system_model}
\end{eqnarray}
where $\rho$ is the pre-precoding ${\sf SNR}$ and ${\bf n}_k \sim {\cal CN}( {\bf 0},
\bI_{ N_{\sf r}})$ is the $N_{\sf r} \times 1$ white Gaussian noise vector
added at the $k$-th user. We assume that ${\bf s}_m$ are i.i.d.\ complex Gaussian
random vectors %from a certain underlying constellation and satisfy
with $\expec[\bs_m] = \mathbf{0}$ and $\expec[{\bf s}_m {\bf s}_m^\dagger] = \bI_{r_m}$.

At the $k$-th user, we assume that ${\bf y}_k$ is processed (down-converted) with an
$N_{\sf r} \times M_{\sf r}$ user-specific RF combiner ${\bf G}_{ {\sf RF}, \hsppp k}$
followed by a user-specific $M_{\sf r} \times r_k$ digital combiner ${\bf G}_{ {\sf Dig},
\hsppp k}$ to produce an estimate of ${\bf s}_k$ as follows %result in
\begin{align}
& \widehat{\bf s}_k =
{\bf G}_{ {\sf Dig}, \hsppp k} ^{\dagger} {\bf G}_{ {\sf RF}, \hsppp k}^{\dagger}
{\bf y}_k
%\nonumber 
\\
& {\hspace{0.15in}}
= %& = &
\sqrt{\frac{\rho}{K}} \hsppp
{\bf G}_{ {\sf Dig}, \hsppp k} ^{\dagger} {\bf G}_{ {\sf RF}, \hsppp k}^{\dagger}
{\bf H}_k {\bf F}_{\sf RF} {\bf F}_{ {\sf Dig},\hsppp k} {\bf s}_k
%\nonumber \\ & {\hspace{0.8in}}
+ \sqrt{\frac{\rho}{K}} \hsppp
{\bf G}_{ {\sf Dig}, \hsppp k} ^{\dagger} {\bf G}_{ {\sf RF}, \hsppp k}^{\dagger}
{\bf H}_k {\bf F}_{\sf RF} \sum_{m = 1, m \neq k}^K
{\bf F}_{ {\sf Dig},\hsppp m} {\bf s}_m + {\bf n}_k.
\label{eq_decoding}
\end{align}
The achievable rate ${\cal R}_k$ (in nats/s/Hz) at the $k$-th user when treating
multi-user interference as noise %with a mismatched decoder~\cite{lapidoth_mismatched}
is given as
\begin{eqnarray}
{\cal R}_k & = & \log \det \left( {\bf I}_{r_k} + \frac{\rho}{K}
{\bf G}_{ {\sf Dig}, \hsppp k} ^{\dagger} {\bf G}_{ {\sf RF}, \hsppp k}^{\dagger}
{\bf H}_k {\bf F}_{\sf RF} {\bf F}_{ {\sf Dig},\hsppp k}
{\bf F}_{ {\sf Dig},\hsppp k}^{\dagger} {\bf F}_{\sf RF}^{\dagger}
{\bf H}_k^{\dagger} {\bf G}_{ {\sf RF}, \hsppp k} {\bf G}_{ {\sf Dig}, \hsppp k}
\cdot {\bf \Sigma}_{\sf intf}^{-1} \right)
\end{eqnarray}
where ${\bf \Sigma}_{\sf intf}$ denotes the interference and noise covariance matrix
\begin{eqnarray}
{\bf \Sigma}_{\sf intf} & = & {\bf I}_{r_k} +
\frac{\rho}{K}
{\bf G}_{ {\sf Dig}, \hsppp k} ^{\dagger} {\bf G}_{ {\sf RF}, \hsppp k}^{\dagger}
{\bf H}_k {\bf F}_{\sf RF} \left( \sum_{m \neq k} {\bf F}_{ {\sf Dig},\hsppp m}
{\bf F}_{ {\sf Dig},\hsppp m}^{\dagger} \right) {\bf F}_{\sf RF}^{\dagger}
{\bf H}_k^{\dagger} {\bf G}_{ {\sf RF}, \hsppp k} {\bf G}_{ {\sf Dig}, \hsppp k}.
\end{eqnarray}

The traditional use of {\em finite-rate feedback} has been to convey the index of
a precoder matrix from an appropriately-designed codebook of precoders to assist
with adaptive transmissions to improve ${\cal R}_k$~\cite{david_review,vasanth_design}.
More generally, feedback from
users can also be used to aid in scheduling, channel estimation and
advanced/non-codebook based precoder design. In this work, as we will see later
in Sec.~\ref{sec3}, we assume that each user feeds back its top beam indices,
an estimate of the received ${\sf SNR}$ and signal phase, and cross-correlation
of the top receive beams to assist with the design of
a non-codebook based multi-user precoder structure.
%following feedback available from users: top beam information, estimated receive $\sf SNR$, estimated receive signal phase, and top beam cross-correlations. The amount of feedback overhead required for each quantity and their associated quantization schemes are described in Section III. We also demonstrate that this feedback is sufficient for estimating $\sf SINR$ and performing advanced beamforming that outperforms a simple beam-steering scheme.
In terms of precoder constraints, we make the assumption that
${\bf F}_{ {\sf Dig}, \hsppp m} \in {\mathbb{C}}^{ M_{\sf t} \times r_m}$.

For the RF precoder, we assume that the amplitude and phase of each entry in ${\bf F}_{\sf RF}$
are controlled by a finite precision gain controller and phase shifter, respectively.
In other words, the amplitude and phase come from a set of $2^{B_{\sf amp}}$ and
$2^{B_{\sf phase}}$ quantization levels
\begin{eqnarray}
|{\bf F}_{\sf RF}(i,j)| \in \left\{ A_1, \cdots, A_{2^{B_{\sf amp}}} \right\},
\hspp \hspp
\angle{ {\bf F}_{\sf RF}(i,j) } \in \left\{ \phi_1, \cdots, \phi_{2^{B_{\sf phase}}} \right\},
\label{eq_hybrid_precoder_constraints}
\end{eqnarray}
where $0 \leq A_1 < A_2 < \dots < A_{2^{B_{\sf amp}}}$. Prior work on hybrid
beamforming such as~\cite{oelayach,alkhateeb,alkhateeb2,sohrabi} etc., assume
that the RF precoder can only be controlled by a phase shifter. However, such
constraining assumptions are not reflective of practical
implementations~\cite{huang2015low,briqech2017low,vasanth_comm_mag_16}, where an
independent gain controller can be used in every RF chain for every antenna.
With these structural constraints on the precoder, the transmit power constraint
is captured by %the equation
\begin{eqnarray}
\sum_{m = 1}^K {\sf Tr} \left( {\bf F}_{ {\sf Dig}, \hsppp m}^{\dagger}
{\bf F}_{\sf RF}^{\dagger} {\bf F}_{\sf RF} {\bf F}_{ {\sf Dig}, \hsppp m} \right) \leq K.
\end{eqnarray}

We are interested in the design of RF and digital precoders with the sum rate,
${\cal R}_{\sf sum} \triangleq \sum_{k = 1}^K {\cal R}_k$, being the metric to maximize.
In general, we only need the constraints $\sum_{k = 1}^{K} r_k \leq M_{\sf t} \leq N_{\sf t}$
and $\max _{k} r_k \leq M_{\sf r} \leq N_{\sf r}$.
%and for simplicity, we assume that $r_k = \frac{ M _{\sf t}}{K}, \hsppp k = 1, \cdots, K$.
However, the considered sum rate optimization with such an assumption is quite
complicated. To overcome this complexity, we %consider %two use-cases
consider a simple use-case %of practical importance
in this work.
%and restrict attention to these two use-cases.

%\section{Proposed Multi-User Beamforming Design}
%\section{Use-Case I: One RF Chain at the User End}
\section{Multi-User Beamformer Design}
\label{sec3}
%\noindent {\bf \em \underline{One RF Chain at User End:}}
%In the first use-case,
We are interested in the practically-motivated setting where each user is equipped
with only one RF chain and the base-station transmits one data-stream to each user that is simultaneously
scheduled. In this scenario, $M_{\sf r} = r_k = 1$ (for all $k = 1, \cdots, K$)
and $M_{\sf t} = K \leq N_{\sf t}$. The system decoding model in~(\ref{eq_system_model})
and~(\ref{eq_decoding}) reduce to
\begin{eqnarray}
\widehat{\bf s}_k =
{\bf G}_{ {\sf Dig}, \hsppp k} ^{\dagger} {\bf G}_{ {\sf RF}, \hsppp k}^{\dagger} {\bf y}_k
& = &
\underbrace{ {\bf G}_{ {\sf Dig}, \hsppp k} ^{\dagger} }_{1 \times 1}
\underbrace{ {\bf G}_{ {\sf RF}, \hsppp k}^{\dagger} }_{1 \times N_{\sf r}}
\cdot \left(
\sqrt{ \frac{\rho}{K}} \hsppp {\bf H}_k \hsppp
\underbrace{ {\bf F}_{\sf RF} }_{ N_{\sf t} \times K}
\cdot \underbrace{ {\bf F}_{\sf Dig} }_{K \times K}
\cdot \underbrace{ {\bf s} }_{K \times 1}  + {\bf n}_k \right) \\
& = & \sqrt{\frac{\rho}{K}} \cdot
{\bf g}_k^{\dagger} \hsppp {\bf H}_k \hsppp \left[ {\bf f}_1 {\bf s}_1 ,
\cdots, {\bf f}_K {\bf s}_K \right] + {\bf g}_k^{\dagger} \hsppp {\bf n}_k
\end{eqnarray}
where ${\bf F}_{\sf Dig} = \left[ {\bf F}_{{\sf Dig}, 1}, \cdots,
{\bf F}_{{\sf Dig},K} \right]$ and ${\bf s} = \left[ {\bf s}_1 , \cdots,
{\bf s}_K \right]^{\sf T}$, %with ${\bf s}_m$
%denoting the complex Gaussian input intended for the $m$-th user,
and the second equation follows assuming\footnote{A simple realization of the
hybrid precoding architecture is achieved by setting ${\bf F}_{\sf Dig} = {\bf I}_K$
and the desired ${\bf f}_k$ for the $k$-th user is set as the $k$-th column of
${\bf F}_{\sf RF}$. The desired ${\bf f}_k$ is such that ${\bf f}_k^{\dagger}
{\bf f}_k \leq 1$ and meets the quantization constraints
in~(\ref{eq_hybrid_precoder_constraints}). In a practical implementation,
${\bf F}_{\sf Dig}$ could be primarily used for sub-band precoding and in the
narrowband context of this work, ${\bf F}_{\sf Dig} = {\bf I}_K$ would reflect
such an implementation-driven model.}
%, and optimizes an appropriately chosen metric.]
${\bf f}_k = {\bf F}_{\sf RF}
{\bf F}_{ {\sf Dig}, k}$ and ${\bf G}_{ {\sf RF}, \hsppp k} = {\bf g}_k$.
%= {\bf I}_K$,  and with %${\bf F}_{{\sf RF}, \hsppp m}$
%${\bf f}_m$ denoting the $m$-th column of ${\bf F}_{\sf RF}$.
The power constraint %on ${\bf F}_{\sf RF}$
%in either scenario
is equivalent to
%\begin{eqnarray}
%{\sf Tr} \big( {\bf F}_{\sf RF}^{\dagger} {\bf F}_{\sf RF} \big)
%=
$\sum_{m = 1}^K  {\bf f}_k^{\dagger} {\bf f}_k \leq K$
%{\bf F}_{{\sf RF}, \hsppp m}^{\dagger}{\bf F}_{{\sf RF}, \hsppp m} = 1.
%\end{eqnarray}
and ${\cal R}_k$ reduces to
\begin{eqnarray}
{\cal R}_k = \log \left( 1 + \frac{ \frac{\rho}{K} \cdot |{\bf g}_k^{\dagger}
{\bf H}_k {\bf f}_k|^2 } {1 + \frac{\rho}{K} \cdot \sum_{m \neq k}
|{\bf g}_k^{\dagger} {\bf H}_k {\bf f}_m|^2 } \right).
\end{eqnarray}

The focus of this section is to first develop an advanced feedback
mechanism and a systematic design of the multi-user beamforming
structure based on a directional representation of the channel. This
structure allows the base-station to combat multi-user interference in simultaneous
transmissions. %Leveraging the insights from this constructive scheme, we then
%establish two upper bounds for ${\cal R}_{\sf sum}$ and compare the constructive
%scheme with the upper bounds.

\subsection{Initial Beam Alignment}
\label{sec3a}
Enabling multi-user transmissions in practice is critically dependent on an
initial beam acquisition process (commonly known as the {\em beam alignment}
phase). In a practical implementation such as 3GPP 5G-NR, beam alignment
corresponds to a beam sweep over a block of secondary synchronization (SS)
signals transmitted over multiple ports/RF chains. The use of multiple
directional beams over multiple ports results in a composite beam pattern
at the base-station end (as seen from the user side). The composite pattern
can lead to uncertainty in the direction of the strongest path between the
base-station and the user. This directional ambiguity is subsequently resolved
with a beam refinement over the individual constituent beams that make the
composite beam on separate resource elements. Beam refinement allows identification
and ambiguity resolution of the constituent beams.

Such a ``post directional ambiguity resolved'' beam alignment process is
modeled by assuming that the base-station is equipped with an $N$ element
codebook ${\cal F}_{\sf tr}$
\begin{eqnarray}
\label{eq:bs_codebook_tx}
{\cal F}_{\sf tr} =
\Big\{ \bff_{{\sf tr},1}, \, \dots, \, \bff_{{\sf tr},N} \Big\},
\end{eqnarray}
and the $k$-th user is equipped with an $M$ element user-specific codebook
${\cal G}_{\sf tr}^k$
\begin{eqnarray}
\label{eq:bs_codebook_rx}
{\cal G}_{\sf tr}^k =
\left\{ \bg^\ind{k}_{{\sf tr},1}, \, \dots, \, \bg^\ind{k}_{{\sf tr},M} \right \}.
\end{eqnarray}
A typical design methodology for ${\cal F}_{\sf tr}$ is a hierarchical design with
different sets of beams that trade-off peak array gain at the cost of initial beam
acquisition latency. For example, at least from the 3GPP 5G-NR perspective, the
designs of ${\cal F}_{\sf tr}$ and ${\cal G}_{\sf tr}^k$ are intended to be
%proprietary
implementation-specific at the base-station and user ends, respectively.
Nevertheless, overarching design guidelines for beam broadening are provided
in~\cite{hur,raghavan_jstsp,song17codebook,noh17adaptive}. In particular, a broadened
beam can be generated by an optimal co-phasing of a number of array steering vectors
in appropriately chosen directions. Both the number of such vectors as well as their
steering directions can be optimized to produce a broadened beam. It must also be
pointed out that most of the beam broadening works have some variations in terms
of design principles and these variations themselves do not affect the flavor of
results reported in this paper.

In the beam alignment phase, the top-$P$ beam indices at the base-station and each
user that maximize an estimate of the received ${\sf SNR}$ are learned. In particular,
the received ${\sf SNR}$
corresponding to the $(m, n)$-th beam index pair at the $k$-th user is given as
\begin{eqnarray}
{\sf SNR}^\ind{k}_{\sf rx}(m,n) =
\abs{\left(\bg^\ind{k}_{{\sf tr},m}\right)^\dagger \bH_k \bff_{{\sf tr},n}}^2.
\end{eqnarray}
Let the beam pair indices at the $k$-th user be arranged in non-increasing
order of the received ${\sf SNR}$ and let the top-$P$ beam pair indices be
denoted as
\begin{eqnarray}
{\cal M} =
\Big\{ \left(m_1^k, \, n_1^k\right), \hspp \cdots, \hspp
\left(m_P^k, \, n_P^k \right) \Big\}.
\end{eqnarray}
With the simplified notation of
\begin{eqnarray}
{\sf SNR}^\ind{k}_{ {\sf rx}, \hsppp \ell} \triangleq
{\sf SNR}^\ind{k}_{\sf rx}(m_{\ell}^{k},n_{\ell}^{k}), \hspp \ell = 1, \cdots, P,
\end{eqnarray}
we have ${\sf SNR}^\ind{k}_{ {\sf rx}, \hsppp 1} \geq \cdots \geq
{\sf SNR}^\ind{k}_{ {\sf rx}, \hsppp P}$.
%That is, ${\sf SNR}^\ind{k}_{\sf rx}(m_1^{k},n_1^{k}) \geq \cdots \geq
%{\sf SNR}^\ind{k}_{\sf rx}(m_P^{k},n_P^{k})$ and we use the simplified
%notation
With the initial beam alignment methodology as described above, %in Sec.~\ref{sec3a},
we now leverage the top-$P$ beam information learned at the $k$-th user
%in the beam alignment phase
to estimate the channel matrix ${\bf H}_k$ and to design ${\bf F}_{\sf RF}$ at
the base-station end.

%\subsection{Reconstructive Beamforming}
%\subsection{Channel Reconstruction}
\subsection{Channel Reconstruction and Beamformer Design}
\label{sec3b}
A typical use of the feedback information at the base-station is to select the
top/best beam indices for all the users and to leverage this information to construct
a multi-user transmission scheme. Such an approach is adopted in~\cite{vasanth_jsac2017},
where multi-user beam designs leveraging only the top beam pair index,
$\left(m_1^k, \, n_1^k\right)$, and intended to serve different objectives are proposed:
i) greedily (from each user's perspective) steering a beam to the best direction for
that user (called the {\em beam steering} scheme), ii) using the information collated
from different users to combat interference to other simultaneously scheduled users via
a zeroforcing solution (called the {\em zeroforcing} scheme), and iii) for leveraging
both the beam steering and interference management objectives via a generalized
eigenvector optimization (called the {\em generalized eigenvector} scheme). If the beam
pair $\left(m_1^k, \, n_1^k\right)$ is blocked or fades, the $k$-th user requests the
base-station to switch to the beam index $n_2^k$ and it switches to the beam with index
$m_2^k$ (and so on)~\cite{vasanth_tap_blockage}.

In this work, we propose to generalize the structures in~\cite{vasanth_jsac2017} by
leveraging {\em all} the top-$P$ beam pair indices fed back from each user. %in these objectives.
In this direction, the base-station intends to {\em reconstruct} or {\em estimate} a rank-$P$
approximation of (a scaled version of) the channel matrix ${\bf H}_k$
corresponding to the $k$-th user as follows
\begin{eqnarray}
\widehat{\bf H}_k = \sum_{\ell = 1}^P \widehat{\alpha}_{k, \ell} \hsppp
\widehat{ \bu}_{k,\ell} \hsppp \widehat{\bv}_{k,\ell}^{\dagger},
\label{eq_channel_reconstruction}
\end{eqnarray}
where $\widehat{ \bu}_{k,\ell}$ and $\widehat{ \bv}_{k,\ell}$ are defined as
estimates of the array steering vectors $\bu_{k, \ell}$ and $\bv_{k,\ell}$, respectively.
Given the channel model structure in~(\ref{eq:channel}),~(\ref{eq_channel_reconstruction})
is simplified by estimating ${\bf v}_{k, \hsppp \ell}$ and $|\alpha_{k, \hsppp \ell}|$ by
${\bf f}_{ {\sf tr}, n_{\ell}^k}$ and $\gamma_{k, \ell}$, respectively, where
\begin{eqnarray}
\gamma_{k, \ell} \triangleq
%{\cal Q}_{ B_{\sf SNR} }
%\left( \sqrt{ {\sf SNR}^\ind{k}_{ {\sf rx}, \hsppp \ell} } \right)
\sqrt{ {\cal Q}_{ B_{\sf SNR} } \left(  {\sf SNR}^\ind{k}_{ {\sf rx}, \hsppp \ell}  \right) }
\end{eqnarray}
%(where $\ell = 1, \cdots, P$), respectively,
for some choice of $B_{\sf SNR}$. In the above description, ${\cal Q}_B(\cdot)$ denotes
an appropriately-defined $B$-bit quantization operation\footnote{A $B$-bit quantization
operation is precisely specified if $2^B$ disjoint intervals that exactly and entirely
span the range
of the quantity and a representative/quantized value from each interval are specified.}
of the quantity under consideration.
%In practice, ${\sf SNR}^\ind{k}_{ {\sf rx}, \hsppp \ell}$
%has to be quantized with a resolution of $B_{\sf SNR}$ bits and fed back.
However, estimating $\widehat{\bf H}_k$ as in~(\ref{eq_channel_reconstruction}) is not
complete until we have an estimate for $\angle{ \alpha_{k,\ell} }$ and
${\bf u}_{k, \ell}$. The quantity $\angle{ \alpha_{k, \ell}} $ can be estimated
by the user with the same reference symbol resource (or pilot symbol) transmitted
during the beam training phase with no additional training overhead. Therefore,
we define $\varphi_{k,\ell}$ as the
$B_{\sf est, \hsppp phase}$-bit quantization of the phase of an estimate
$\widehat{\bf s}_{{\sf tr}, k,\ell}$ of the pilot symbol
${\bf s}_{{\sf tr}, k,\ell}$ % as below:
\begin{eqnarray}
\varphi_{k,\ell} \triangleq {\cal Q}_{B_{\sf est, \hsppp phase}}
\left( \angle{ \widehat{\bf s}_{{\sf tr}, k,\ell} } \right), \hspp \hspp
%%%%
%%%%
{\sf where} \hspp \hspp
\widehat{\bf s}_{{\sf tr}, k,\ell} = \left( {\bf g}^{(k)}_{{\sf tr}, \hsppp m_{\ell}^k } \right)^{\dagger}
\left[ \sqrt{\rho} \hsppp {\bf H}_k {\bf f}_{{\sf tr}, \hsppp n_{\ell}^k }
{\bf s}_{{\sf tr}, k,\ell}
+ %\widetilde
{\bf n}_{k,\ell} \right]
\end{eqnarray}
for some choice of $B_{\sf est, \hsppp phase}$. The noise term %$\widetilde{\bf n}$
${\bf n}_{k,\ell}$ captures the additive noise in the initial beam alignment
process corresponding to the top-$P$ beam pairs.
%where ${\cal Q}(\cdot)$ denotes the $B_{\sf phase}$-bit quantization operation
%of the phase quantity under consideration.

For ${\bf u}_{k, \ell}$, %the latter quantity,
we note that the base-station not only needs the beam indices
$\{ m_{\ell}^k \}$ that are useful for the user side, but also the useful part of the
user's codebook (${\cal G}_{\sf tr}^k$) since the base-station is typically unaware
of it. %the user's codebook.
To avoid this unnecessary complexity and feedback given the proprietary structure of
${\cal G}_{\sf tr}^k$, we assume that the $k$-th user uses a multi-user reception beam
${\bf g}_k$. In the simplest manifestation, ${\bf g}_k$ could be the best training
beam learned in the beam alignment phase, $\bg_{{\sf tr}, m_1^{k}}^{(k)}$.
However, a more sophisticated choice for ${\bf g}_k$ is not precluded. For example, an
iterative choice that maximizes the ${\sf SINR}$ (instead of the ${\sf SNR}$) could
be considered for ${\bf g}_k$.

We then note that the estimated ${\sf SINR}$, defined as,
\begin{eqnarray}
\widehat{\sf SINR}_k \triangleq
\frac{ \frac{\rho}{K} \cdot | {\bf g}_k^{\dagger} \widehat{\bf H}_k {\bf f}_k |^2 }
{ 1 + \frac{ \rho}{K} \cdot \sum_{m \neq k} | {\bf g}_k^{\dagger} \widehat{\bf H}_k {\bf f}_m |^2 }
\label{eq_estimated_SINR}
\end{eqnarray}
is only dependent on $\widehat{\bf H}_k$ in the form of ${\bf g}_k^{\dagger} \widehat{\bf H}_k$.
Building on this fact, each user %obtains $P$ scalar quantities,
generates $\{ \beta_{k, \hsppp \ell} \}$, defined as,
\begin{eqnarray}
%\widehat{\bu}_{k,\ell} & = & {\bf g}_{ {\sf tr}, \hsppp m_{\ell}^k}^{(k)} \\
\beta_{k, \hsppp \ell} \triangleq
%& \triangleq &
{\bf g}_k^{\dagger} \widehat{\bu}_{k,\ell} \hsppp \hspp
{\sf where} \hspp \hsppp
\widehat{\bu}_{k,\ell} = {\bf g}_{ {\sf tr}, \hsppp
m_{\ell}^k}^{(k)}.
%{\bf g}_{ {\sf tr}, \hsppp m_{\ell}^k}^{(k)}
\end{eqnarray}
It then quantizes the amplitude and phase of $\beta_{k, \ell}$ for some choice
of $B_{\sf corr, \hsppp amp}$ and $B_{\sf corr, \hsppp phase}$ and feeds them back
\begin{eqnarray}
\mu_{k, \ell} \triangleq {\cal Q}_{ B_{\sf corr,\hsppp amp}} \left( |\beta_{k,\ell}| \right),
\hspp \hspp
\nu_{k,\ell} \triangleq {\cal Q}_{ B_{\sf corr,\hsppp phase}} \left( \angle{ \beta_{k,\ell} }
\right).
\end{eqnarray}
For both $\varphi_{k,\ell}$ and $\nu_{k,\ell}$, without loss in generality, relative
phases with respect to $\varphi_{k,1}$ and $\nu_{k,1}$ (that is,
$\varphi_{k,\ell} - \varphi_{k,1}$ and $\nu_{k,\ell} - \nu_{k,1}$) can be reported.

The mappings between the quantities of interest and the approximated quantities as
well as the feedback overhead needed from each user to implement the proposed scheme
are described in Table~\ref{table_feedback}.
%Note that since the estimated SINR is invariant to constant phase offsets, both
%$\varphi_{k,1}$ and $\nu_{k,1}$ can be assumed to be 0 by setting those quantities
%as the respective reference phases, which reduced the overhead for those quantities by $P$.
While the feedback overhead increases
linearly with $P$ (the rank of the channel approximation), there are diminishing
returns in terms of channel representation accuracy since the clusters captured
in $\widehat{\bf H}_k$ are sub-dominant as $P$ increases (and are eventually limited by
$L_k$). Thus, it is useful to select $P$ to trade-off these two conflicting objectives.

\begin{table*}[htb!]
\caption{Mappings between quantities describing ${\bf H}_k$ and
the approximated quantities, and their feedback overhead.}
\label{table_feedback}
\begin{center}
\begin{tabular}{|c||c||c|}
\hline
Quantity of Interest & Approximated Quantity & Feedback Overhead \\
\hline
\hline
Array steering vector at base-station end (${\bf v}_{k, \ell}$) &
Base-station beam indices ($n_{\ell}^k$) &
$P \cdot \log_2(N)$ \\ \hline
%%%%
Gain of cluster coefficient ($|\alpha_{k, \ell}|$) &
Received ${\sf SNR}$ in beam alignment (${\sf SNR}_{ {\sf rx}, \hsppp \ell }^{(k)}$) &
$P \cdot B_{\sf SNR}$ \\ \hline
%%%%
Phase of cluster coefficient ($\angle{\alpha_{k,\ell}}$) &
%Phase of cluster coefficient
Estimated phase in beam alignment ($\angle{{\widehat{\bf s}}_{{\sf tr} , k, \ell} }$)
& %($\varphi_{k, \ell}$) &
$(P-1) \cdot B_{\sf est, \hsppp phase}$ \\ \hline
%%%%
Array steering vector at user end (${\bf u}_{k, \ell}$) &
Amplitude of codebook correlation %amplitude
($|\beta_{k, \ell}|$) &
$P \cdot B_{\sf corr, \hsppp amp}$
\\
%\cline{1-2}
& Phase of codebook correlation %phase
($\angle{\beta_{k,\ell}}$) &
$(P-1) \cdot B_{\sf corr, \hsppp phase}$
\\ \hline
%%%%
%\hline
%\hline
\end{tabular}
\end{center}
\end{table*}
%\end{center}

Following the above discussion, the $k$-th user feeds back the $P \times 5$ matrix
${\bf P}_k$, defined as
\begin{eqnarray}
{\bf P}_k \triangleq \left[
\begin{array}{ccccc}
n_{1}^k & \gamma_{k, 1} & 0 & \mu_{k,1} & 0 \\
n_{2}^k & \gamma_{k, 2} & \varphi_{k, 2} - \varphi_{k, 1} & \mu_{k,2}
& \nu_{k, 2} - \nu_{k, 1} \\
\vdots & \vdots & \vdots & \vdots & \vdots \\
n_{P}^k & \gamma_{k, P} & \varphi_{k, P} - \varphi_{k,1} & \mu_{k, P} & \nu_{k, P} - \nu_{k,1}
%\vdots & \vdots & \vdots & \vdots & \vdots \\
\end{array}
\right],
\end{eqnarray}
%$n_{\ell}^k$, $\gamma_{k, \ell}$, $\varphi_{k, \ell}$, $\mu_{k,\ell}$ and $\nu_{k, \ell}$
%%quantized versions of $\{ {\sf SNR}_{ {\sf rx}, \hsppp \ell }^{(k)} \}$
%%and $\{ \beta_{k, \ell} \}$
%for $\ell = 1, \cdots, P$,
and the base-station approximates ${\bf g}_k^{\dagger} \widehat{\bf H}_k$ %at the $k$-th user
as follows
\begin{eqnarray}
{\bf g}_k^{\dagger} \widehat{\bf H}_k = \sum_{\ell = 1}^P
%\sqrt{ {\sf SNR}^\ind{k}_{ {\sf rx}, \hsppp \ell} }
\mu_{k,\ell} \hsppp \gamma_{k,\ell} \cdot
e^{j (\varphi_{k, \ell} + \nu_{k,\ell} )}
\cdot \left( {\bf f}_{ {\sf tr}, n_{\ell}^k} \right)^{\dagger}.
\label{eq_eigenspace_reconst}
\end{eqnarray}
In other words, ${\bf g}_k^{\dagger} \widehat{\bf H}_k$ is represented as a linear
combination of the top-$P$ beams as estimated from ${\cal F}_{\sf tr}$ in the initial
beam alignment phase. The weights in this linear combination correspond to the relative
strengths of the clusters as distinguished by the codebook resolution (at both ends).

%\subsection{Beamformer Design}
%\label{sec3c}
The base-station uses the channel matrix constructed for each user based on its feedback
information (${\bf g}_k^{\dagger} \widehat{\bf H}_k$) and generates a good
beamformer structure, illustrated in the next result, for use in multi-user
transmissions.
\begin{prop}
\label{prop_zf}
The zeroforcing beamformer structure is one where for every user that is simultaneously
scheduled, the beam ${\bf f}_k$ nulls the multi-user interference in
$\widehat{\sf SINR}_m, \hspp m \neq k$ with $\widehat{\sf SINR}_m$ as given
in~(\ref{eq_estimated_SINR}). The beams $\{ {\bf f}_m \}$ in the zeroforcing
structure are the unit-norm column vectors of the $N_{\sf t} \times K$ matrix
${\cal H}^{\dagger} \left( {\cal H} {\cal H}^{\dagger} \right)^{-1}$, where
${\cal H}$ is the $K \times N_{\sf t}$ matrix given as
\begin{eqnarray}
{\cal H} = \left[
\begin{array}{c}
{\bf g}_1^{\dagger} \widehat{\bf H}_1 \\
{\bf g}_2^{\dagger} \widehat{\bf H}_2 \\
\vdots \\
{\bf g}_K^{\dagger} \widehat{\bf H}_K
\end{array}
\right] =
\left[
\begin{array}{c}
\sum_{\ell = 1}^P \mu_{1,\ell} \hsppp \gamma_{1,\ell} \cdot
e^{j (\varphi_{1, \ell} + \nu_{1,\ell} )}
\cdot \left( {\bf f}_{ {\sf tr}, n_{\ell}^1} \right)^{\dagger}
\\
\sum_{\ell = 1}^P \mu_{2,\ell} \hsppp \gamma_{2,\ell} \cdot
e^{j (\varphi_{2, \ell} + \nu_{2,\ell} )}
\cdot \left( {\bf f}_{ {\sf tr}, n_{\ell}^2} \right)^{\dagger}
\\
\vdots \\
\sum_{\ell = 1}^P \mu_{K,\ell} \hsppp \gamma_{K,\ell} \cdot
e^{j (\varphi_{K, \ell} + \nu_{K,\ell} )}
\cdot \left( {\bf f}_{ {\sf tr}, n_{\ell}^K} \right)^{\dagger}
\end{array}
\right].
\end{eqnarray}
\end{prop}
\begin{proof}
See Appendices~\ref{app_ge} and~\ref{proof_prop_zf}.
\end{proof}

\section{Upper Bounds for ${\cal R}_{\sf sum}$}
\label{sec4}
We are interested in benchmarking the performance of the zeroforcing structure
against an upper bound on ${\cal R}_{\sf sum}$. The goal of optimizing
${\cal R}_{\sf sum}$ over $\{ {\bf f}_k, \hsppp {\bf g}_k \}$ with perfect channel
state information $\{ {\bf H}_k \}$ is a non-convex optimization
problem~\cite{cioffi_weighted,kobayashi,vasanth_arxiv2011}
that appears to be complicated. In this context, an alternate formulation based on the
signal-to-leakage and noise ratio metric~\cite{tarighat} that {\em simultaneously} maximizes the
array gain seen by the $k$-th user, $| {\bf g}_k^{\dagger} {\bf H}_k {\bf f}_k |^2$,
and minimizes the interfering array gain seen by the other users,
$| {\bf g}_m^{\dagger} {\bf H}_m {\bf f}_k |^2, \hspp m \neq k$ is relevant. Since these objectives
are in some sense conflicting and can be weighed differently, we consider the composite metric
\begin{eqnarray}
{\sf SLNR} _k \triangleq
\frac{ \eta_{k,k} \hsppp | {\bf g}_k^{\dagger} {\bf H}_k {\bf f}_k |^2 }
{ 1 + \sum_{m \neq k} \eta_{m,k} \hsppp
| {\bf g}_m^{\dagger} {\bf H}_m {\bf f}_k |^2 }
\label{eq_SLNR}
\end{eqnarray}
for an appropriate set of weighting factors $\eta_{m,k} \geq 0$ with
$m ,k \in \{ 1, \cdots, K \}$.

\subsection{Upper Bound Motivated by the Zeroforcing Structure}
\label{sec4a}
Building on Prop.~\ref{prop_zf}, we now develop an upper bound for ${\cal R}_{\sf sum}$
motivated by the zeroforcing structure. In this direction, we consider a
signal-to-leakage-type metric equivalent of~(\ref{eq_SLNR}) based on the estimated
channel matrix $\widehat{\bf H}_k$
%that {\em simultaneously} maximizes the array gain
%seen by the $k$-th user, $| {\bf g}_k^{\dagger} \widehat{\bf H}_k {\bf f}_k |^2$,
%and minimizes the interfering array gain seen by the other users,
%$| {\bf g}_m^{\dagger} \widehat{\bf H}_m {\bf f}_k |^2, \hspp m \neq k$. Since
%these objectives are in some sense conflicting and can be weighed differently, we
%consider the following composite metric:
\begin{eqnarray}
\widehat{ {\sf SLNR} }_k \triangleq
\frac{ \eta_{k,k} \hsppp | {\bf g}_k^{\dagger} \widehat{\bf H}_k {\bf f}_k |^2 }
{ 1 + \sum_{m \neq k} \eta_{m,k} \hsppp
| {\bf g}_m^{\dagger} \widehat{\bf H}_m {\bf f}_k |^2 }
\end{eqnarray}
for an appropriate set of weighting factors $\eta_{m,k} \geq 0$ with
$m ,k \in \{ 1, \cdots, K \}$.
\begin{prop}
\label{prop_ge}
%For any $k = 1, \cdots, K$,
Assuming that $\{ \widehat{\bf H}_m^{\dagger} {\bf g}_m \}$ and
$\{ \eta_{m,k} \}$ are known at the base-station, the choice of ${\bf f}_k$
that maximizes $\widehat{ {\sf SLNR}} _k$ is given %as the unit-norm version of
by the generalized eigenvector structure
\begin{eqnarray}
{\bf f}_k = \frac{ \left( {\bf I}_{ N_{\sf t}} + \sum_{m \neq k} \eta_{m,k}
\hsppp
\widehat{\bf H}_m^{\dagger} {\bf g}_m {\bf g}_m^{\dagger} \widehat{\bf H}_m
\right)^{-1} \widehat{\bf H}_k^{\dagger} {\bf g}_k
}
{ \Big\| \left( {\bf I}_{ N_{\sf t}} + \sum_{m \neq k} \eta_{m,k} \hsppp
\widehat{\bf H}_m^{\dagger} {\bf g}_m {\bf g}_m^{\dagger}
\widehat{\bf H}_m \right)^{-1} \widehat{\bf H}_k^{\dagger} {\bf g}_k \Big\| }.
\label{eq_ge}
\end{eqnarray}
\end{prop}
\begin{proof}
See Appendix~\ref{proof_prop_ge}.
\end{proof}

Several remarks are in order at this stage.

\begin{itemize}
\item
In the case where $\eta_{m,k}$ are set to zero for all $m \neq k$ (that is,
the focus is {\em not} on interference management), the solution in~(\ref{eq_ge}) reduces
to
\begin{eqnarray}
{\bf f}_k = \frac{ \widehat{\bf H}_k^{\dagger} {\bf g}_k}
{ \| \widehat{\bf H}_k^{\dagger} {\bf g}_k \| }
= \frac{ \sum_{\ell = 1}^P \mu_{k,\ell} \hsppp \gamma_{k,\ell} \cdot
e^{-j (\varphi_{k, \ell} + \nu_{k,\ell} )}
\cdot {\bf f}_{ {\sf tr}, n_{\ell}^k} }
{ \big\| \sum_{\ell = 1}^P \mu_{k,\ell} \hsppp \gamma_{k,\ell} \cdot
e^{-j (\varphi_{k, \ell} + \nu_{k,\ell} )}
\cdot {\bf f}_{ {\sf tr}, n_{\ell}^k} \big\| }.
\end{eqnarray}
This is not surprising, and the base-station {\em greedily} steers a beam along the
weighted set of top-$P$ beams from ${\cal F}_{\sf tr}$ for the $k$-th user. In other
words, the base-station generates a set of transmit weights that are matched to the
transmit angular spread of the channel as identified by the resolution of ${\cal F}_{\sf tr}$.

\item
In the case where $\eta_{m,k} = 0$ except if $m = k$ or $m = m'$ (for a specific
$m' \neq k$), it can be seen that ${\bf f}_k$ reduces to
\begin{eqnarray}
{\bf f}_k = \frac{ \widehat{\bf H}_k^{\dagger} {\bf g}_k -
\eta_{m',k} \cdot \left( {\bf g}_{m'}^{\dagger} \widehat{\bf H}_{m'}
\widehat{\bf H}_k^{\dagger} {\bf g}_k \right) \cdot
\widehat{\bf H}_{m'}^{\dagger} {\bf g}_{m'}  }
{ \big\| \widehat{\bf H}_k^{\dagger} {\bf g}_k -
\eta_{m',k} \cdot \left( {\bf g}_{m'}^{\dagger} \widehat{\bf H}_{m'}
\widehat{\bf H}_k^{\dagger} {\bf g}_k \right) \cdot
\widehat{\bf H}_{m'}^{\dagger} {\bf g}_{m'}
\big\| }.
\label{eq_fk_case2}
\end{eqnarray}
In other words, the specific design of ${\bf f}_k$ in~(\ref{eq_fk_case2})
removes a certain component of the beam corresponding to the $m'$-th user
from the beam corresponding to the $k$-th user.

%In the scenario where only $\eta_{k,k}$ and $\eta_{m',k}$ are non-zero for some specific
%$m'$ (with $\eta_{m,k} = 0$ provided
\item
In the general case, while it gets much harder to simplify ${\bf f}_k$ in~(\ref{eq_ge}),
%and provide closed-form expressions for it. However, it is easy to
it can be seen that ${\bf f}_k$ has the structure
\begin{eqnarray}
{\bf f}_k = \frac{  \sum_{m = 1}^K \widehat{ \delta}_{m,k} \widehat{\bf H}_m^{\dagger} {\bf g}_m
} { \big\| \sum_{m = 1}^K \widehat{\delta}_{m,k} \widehat{\bf H}_m^{\dagger} {\bf g}_m
\big\| }
\label{eq_ge_structure}
\end{eqnarray}
for some complex scalars $\widehat{\delta}_{m,k}$. In other words, the optimal ${\bf f}_k$ is in the
span of $\{ \widehat{\bf H}_m^{\dagger} {\bf g}_m \}$ with the weights $\{ \widehat{\delta}_{m,k} \}$
that make the linear combination being a complicated function of $\{ \eta_{m,k} \}$ as well as
$\{ \widehat{\bf H}_m^{\dagger} {\bf g}_m \}$.

%\ignore{
\item
The above observations are not entirely surprising given the Karhunen-Lo{\`e}ve
interpretation of the eigen-space of the channel(s)~\cite{akbar,tulino_ind,vasanth_it2}
and utilizing an expansion of ${\bf f}_k$ on this basis. Such an expansion is also consistent
with Prop.~\ref{prop_zf} which shows that in the pure interference management case
($\eta_{m,k} \rightarrow \infty$ for all $m \neq k$), ${\bf f}_k$ is given as
\begin{eqnarray}
{\bf f}_k = \frac{ \sum_{m = 1}^K {\cal G}_{m,k} \widehat{\bf H}_m^{\dagger} {\bf g}_m }
{ \big\| \sum_{m = 1}^K {\cal G}_{m,k} \widehat{\bf H}_m^{\dagger} {\bf g}_m
\big\|}
\label{eq_purezf}
\end{eqnarray}
where the $K \times K$ matrix ${\cal G} = \left( {\cal H} {\cal H}^{\dagger} \right)^{-1}$.
%}

\item
On the other hand, from~(\ref{eq_eigenspace_reconst}), we note that %each
$\widehat{\bf H}_m^{\dagger} {\bf g}_m$ is itself a linear combination of the beams
from ${\cal F}_{\sf tr}$. Thus, ${\bf f}_k$ in~(\ref{eq_ge}) %for all the users
%are in general
is a linear combination of beams from ${\cal F}_{\sf tr}$. In other words, the
design of ${\bf f}_k$ is equivalent to a search over $N$
scalar (complex) weights, where $N$ denotes the size of the initial beam
alignment codebook at the base-station end.
\end{itemize}

%\item
With this interpretation, while Prop.~\ref{prop_ge} considers only the
maximization of $\widehat{ {\sf SLNR}} _k$ (not even the sum rate with
$\widehat{\bf H}_k$), we can consider the optimization of ${\cal R}_{\sf sum}$
over ${\bf f}_k$ from a class ${\cal F}_k$, defined as
\begin{eqnarray}
{\cal F}_k \triangleq \left\{ {\bf f}_k \hspp : \hspp
{\bf f}_k =
\frac{  \sum_{n = 1}^{N} %|{\cal F}_{\sf tr}|}
\delta_{n,k} {\bf f}_{ {\sf tr}, \hsppp n}
} { \big\| \sum_{n = 1}^{N} % |{\cal F}_{\sf tr}| }
\delta_{n,k} {\bf f}_{ {\sf tr}, \hsppp n}
\big\| } \hspp {\sf such} \hspp {\sf that} \hspp \delta_{n,k} \in {\mathbb{C}}, \hspp
k = 1, \cdots, K  \right\}.
\label{eq_class_Fk}
\end{eqnarray}
\begin{thm}
Assume that the same multi-user beams ${\bf g}_k$
as in the zeroforcing scheme are used for reception at the $k$-th user. Let
$\{ \delta^{\star}_{n,k} \}$ be defined as the solution to the search
over the complex scalars $\{ \delta_{n,k} \}$
\begin{eqnarray}
\{ \delta^{\star}_{n,k} \} = \arg\max \limits_{ \{ \delta_{n,k} \hsppp : \hsppp
{\bf f}_k \hsppp \in \hsppp {\cal F}_k \} } {\cal R}_{\sf sum}.
\end{eqnarray}
With ${\bf g}_k$ as above and
\begin{eqnarray}
{\bf f}_k = \frac{  \sum_{n = 1}^{N} %|{\cal F}_{\sf tr}|}
\delta^{\star}_{n,k} {\bf f}_{ {\sf tr}, \hsppp n}
} { \big\| \sum_{n = 1}^{N} % |{\cal F}_{\sf tr}| }
\delta^{\star}_{n,k} {\bf f}_{ {\sf tr}, \hsppp n} \big\| },
\label{eq_fk_scalar}
\end{eqnarray}
we obtain an upper bound to
the sum rate with the zeroforcing scheme.
%, which is denoted as $\overline{\cal R}_{\sf sum}$.
\qed
\end{thm}
The proof is trivial following the structure of ${\bf f}_k$ in the zeroforcing
scheme in~(\ref{eq_purezf}) and the definition of the class ${\cal F}_k$
in~(\ref{eq_class_Fk}).
%In other words, assuming ${\bf g}_k = {\bf g}_{ {\sf tr}, m_1^k}^{(k)}$
%for all $k = 1, \cdots, K$,
Since the structure in~(\ref{eq_fk_scalar}) is obtained as a search over scalar
parameters, we call this upper bound a scalar optimization-based upper bound.
Further, while~(\ref{eq_fk_scalar}) is difficult to practically implement,
it provides a benchmark to compare the realizable zeroforcing scheme of
Prop.~\ref{prop_zf}.

%\item
%Further, while this alternate search (over scalar coefficients) is not %fundamentally
%necessarily simpler, a
Another important consequence of~(\ref{eq_fk_scalar}) %this observation
is that the coefficients of ${\bf f}_k$ for either the zeroforcing or the upper bound
are (in general) not of equal amplitude. Thus, ${\bf f}_k$ has to be quantized for
implementation to ensure that the RF beamforming constraints are satisfied. In
particular, we compute $\widehat{\bf f}_k$ with an appropriate quantization scheme
as below
\begin{eqnarray}
|\widehat{\bf f}_k(i)| = \widetilde{\cal Q}_{ B_{\sf amp}} \left( |{\bf f}_k(i)| \right),
\hspp  \hspp \angle{ \widehat{\bf f}_k(i) } = \widetilde{\cal Q}_{ B_{\sf phase} } \left(
\angle{ {\bf f}_k(i) } \right),
\end{eqnarray}
and use them in transmissions for the $k$-th user. Good choices 
for $\widetilde{\cal Q}(\cdot)$ will be discussed in Sec.~\ref{sec5b}.

\subsection{Bounding ${\cal R}_{\sf sum}$ with an
Alternating/Iterative Optimization}
\label{sec3b}
We now propose an iterative %algorithm for establishing an upper bound to
maximization algorithm to optimize ${\cal R}_{\sf sum}$ over $\{ {\bf f}_k,
{\bf g}_k \}$. In this approach, we first optimize the ${\sf SLNR}$ metric
over ${\bf f}_k$ (assuming ${\bf g}_k$ is fixed), and then optimize the
${\sf SINR}$ metric over ${\bf g}_k$ (assuming ${\bf f}_k$ is fixed). The
algorithm is as follows:
\begin{enumerate}
\item Initialize $\{ {\bf g}_k^{(1)}, \hspp k = 1, \cdots, K \}$ randomly.

\item For $i = 1, \cdots, N_{\sf stop}$, where $N_{\sf stop}$ is chosen
according to a stopping criterion to determine convergence:
\begin{itemize}
\item
With $\{ {\bf g}_k = {\bf g}_k^{(i)}, k = 1, \cdots, K \}$ fixed, compute
${\bf f}_k^{(i)}$ %, \hspp k = 1, \cdots, K \}$
as the solution to the following optimization
\begin{eqnarray}
{\bf f}_k^{(i)} = \arg \max \limits_{ {\bf f}_k }
\max_{ \{ \eta_{m,k} \} } {\sf SLNR}_k.
\end{eqnarray}
From Lemma~\ref{lem_ge_soln} in Appendix~\ref{app_ge}, the solution to the
above problem with $\{ \eta_{m,k} \}$ fixed can be seen to be
\begin{eqnarray}
{\bf f}_k =
\frac{ \Big( {\bf I}_{N_{\sf t}} + \sum_{m \neq k} \eta_{m,k} \hsppp
{\bf H}_m^{\dagger} {\bf g}_m^{(i)} {\bf g}_m^{(i) \hsppp \dagger}
{\bf H}_m  \Big)^{-1} {\bf H}_k^{\dagger} {\bf g}_{k}^{(i)} }
{ \Big\|
\Big( {\bf I}_{N_{\sf t}} + \sum_{m \neq k} \eta_{m,k} \hsppp
{\bf H}_m^{\dagger} {\bf g}_m^{(i)} {\bf g}_m^{(i) \hsppp \dagger}
{\bf H}_m  \Big)^{-1} {\bf H}_k^{\dagger} {\bf g}_{k}^{(i)} \Big\| }.
\end{eqnarray}
This candidate ${\bf f}_k$ has to be used to compute ${\sf SLNR}_k$ for
all possible weights $\{ \eta_{m,k} \}$ and optimized to produce $\{ {\bf f}_k^{(i)} \}$.

\item
With $\{ {\bf f}_k = {\bf f}_k^{(i)}, k = 1, \cdots, K \}$ fixed,
compute ${\bf g}_k^{(i+1)}$ as the solution to the following optimization
\begin{eqnarray}
{\bf g}_k^{(i+1)} = \arg \max \limits_{ {\bf g}_k } {\sf SINR}_k.
\end{eqnarray}
Again, from Lemma~\ref{lem_ge_soln} in Appendix~\ref{app_ge}, we have
\begin{eqnarray}
{\bf g}_k^{(i+1)} = \left( {\bf I}_{N_{\sf r}} +
\frac{\rho}{K} \sum_{m \neq k} {\bf H}_k {\bf f}_m^{(i)} {\bf f}_m^{(i), \dagger}
{\bf H}_k^{\dagger} \right)^{-1} {\bf H}_k {\bf f}_k^{(i)}.
\end{eqnarray}
\end{itemize}

\item Compute ${\cal R}_{\sf sum}$ with $\{ {\bf f}_k^{(N_{\sf stop})} \}$
and $\{  {\bf g}_k^{ (N_{\sf stop} + 1)} \}$ for a (potential) upper bound.
\end{enumerate}
Numerical studies show that for almost all channel realizations, the proposed
algorithm converges in a small number of steps ($N_{\sf stop} \approx 10$) to
lead to a tolerable level of difference between successive iterates of
${\cal R}_{\sf sum}$. Further, while we are unable to theoretically establish
that the proposed algorithm results in an upper bound to ${\cal R}_{\sf sum}$,
numerical studies (see Sec.~\ref{sec5c}) suggest that it leads to an upper bound
for almost all channel realizations.

%\subsection{Extension to More RF Chains at the User End}
%\label{sec4c}
%We now discuss the case in which each user is equipped with $M_r > 1$ RF chains
%and study the structure of the proposed scheme in this case. We assume that we still
%have $M_t = 1$ at the base station.

%\section{Simulation Results}
\section{Numerical Studies}
\label{sec5}
We now present numerical studies in a single-cell downlink
framework to illustrate the advantages of the proposed beamforming solutions.
The channel model from~(\ref{eq:channel}) is used to generate a channel matrix 
with $L_k = 6$ clusters, AoDs uniformly distributed in a $120^{\sf o} \times 
30^{\sf o}$ coverage area, and AoAs uniformly distributed in a 
$120^{\sf o} \times 120^{\sf o}$ coverage area for each of the 
$k = 1, \cdots, K_{\sf cell}$ users in the cell. The AoD spread captures a 
traditional three-sector approach with a $30^{\sf o}$ zenith coverage and the 
AoA spread corresponds to the assumption of the use of multiple 
subarrays~\cite{vasanth_tap2018} with the best subarray limited to a 
$120^{\sf o} \times 120^{\sf o}$ coverage. $L_k = 6$ is justified from 
millimeter wave channel measurements reported in~\cite{vasanth_tap2018,qualcomm2}.
The antenna
dimensions assumed in these studies are $N_{\sf tx} = 16$ and $N_{\sf tz} = 4$ at
the base-station end, and $N_{\sf rx} = 2$ and $N_{\sf rz} = 2$ at each user. We 
consider simultaneous transmissions from the base-station to $K = 2$ out of 
the $K_{\sf cell}$ users in the cell. 

In terms of user scheduling, commonly used
criteria include a round robin or a proportionate fair scheduler. On the other hand,
a recently proposed directional scheduler~\cite{vasanth_jsac2017} leverages the
smaller beamwidths afforded by large antenna dimensions to schedule users with
dominant clusters that are spatially well-separated. In this work, the first of the
$K = 2$ users is scheduled randomly and the second user is chosen to ensure that
${\bf f}_{ {\sf tr}, n_1^2} \neq {\bf f}_{ {\sf tr}, n_1^1}$. In other words, the
considered scheduler implements a {\em directional avoidance} protocol with the
dominant cluster in the channel of the first user separated spatially from the
dominant cluster in the channel of the second user, as parsed by ${\cal F}_{\sf tr}$.
With this scheduler, we now primarily focus on the beamforming aspects.

For the initial beam alignment codebooks, based on the beam broadening principles
proposed in~\cite{raghavan_jstsp}, Figs.~\ref{fig_beampatterns_codebooks}(a)-(d)
%Figs.~\ref{fig_beampatterns_codebooks}(a)-(b)
illustrate the beam patterns in the azimuth plane for codebooks of sizes
$N = 32$, $N = 16$, $N = 8$ and $N = 4$, respectively,
%$N = 8$ and $N = 4$, respectively,
to cover the $120^{\sf o} \times
30^{\sf o}$ AoD space with a $16 \times 4$ planar array at the base-station side.
The optimization proposed in~\cite{raghavan_jstsp} results in a discrete
Fourier transform (DFT) codebook solution for $N = 32$ and $N = 16$.
%and are hence not shown here.
From
Fig.~\ref{fig_beampatterns_codebooks}, we observe that a beam codebook
of small size (e.g., $N = 4$) where each beam offers a broad directional
coverage can reduce the acquisition latency at the cost of peak and/or worst-case
array gain. On the other hand, a beam codebook of large size (e.g., $N = 32$)
where each beam can offer precision in terms of beamspace (and array gain)
comes at the cost of acquisition latency. For the codebooks at the user end,
two codebook sizes ($M = 4$ for a reduced acquisition latency and $M = 16$ for
performance improvement at the cost of acquisition latency) are considered with
similar beam design principles as for the base-station side.

\begin{figure*}[htb!]
\begin{center}
\begin{tabular}{cc}
\includegraphics[height=2.5in,width=3.1in] {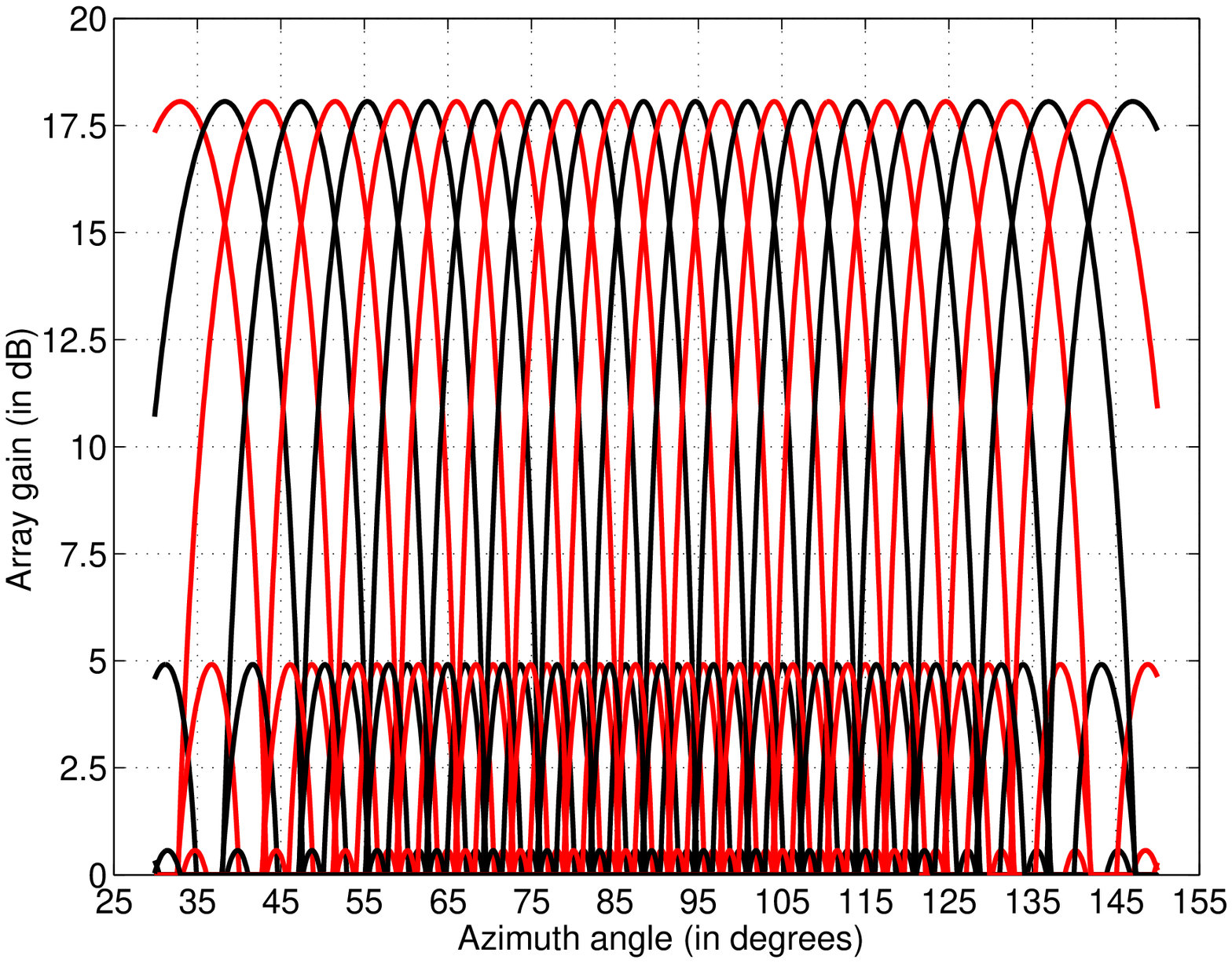}
&
\includegraphics[height=2.5in,width=3.1in] {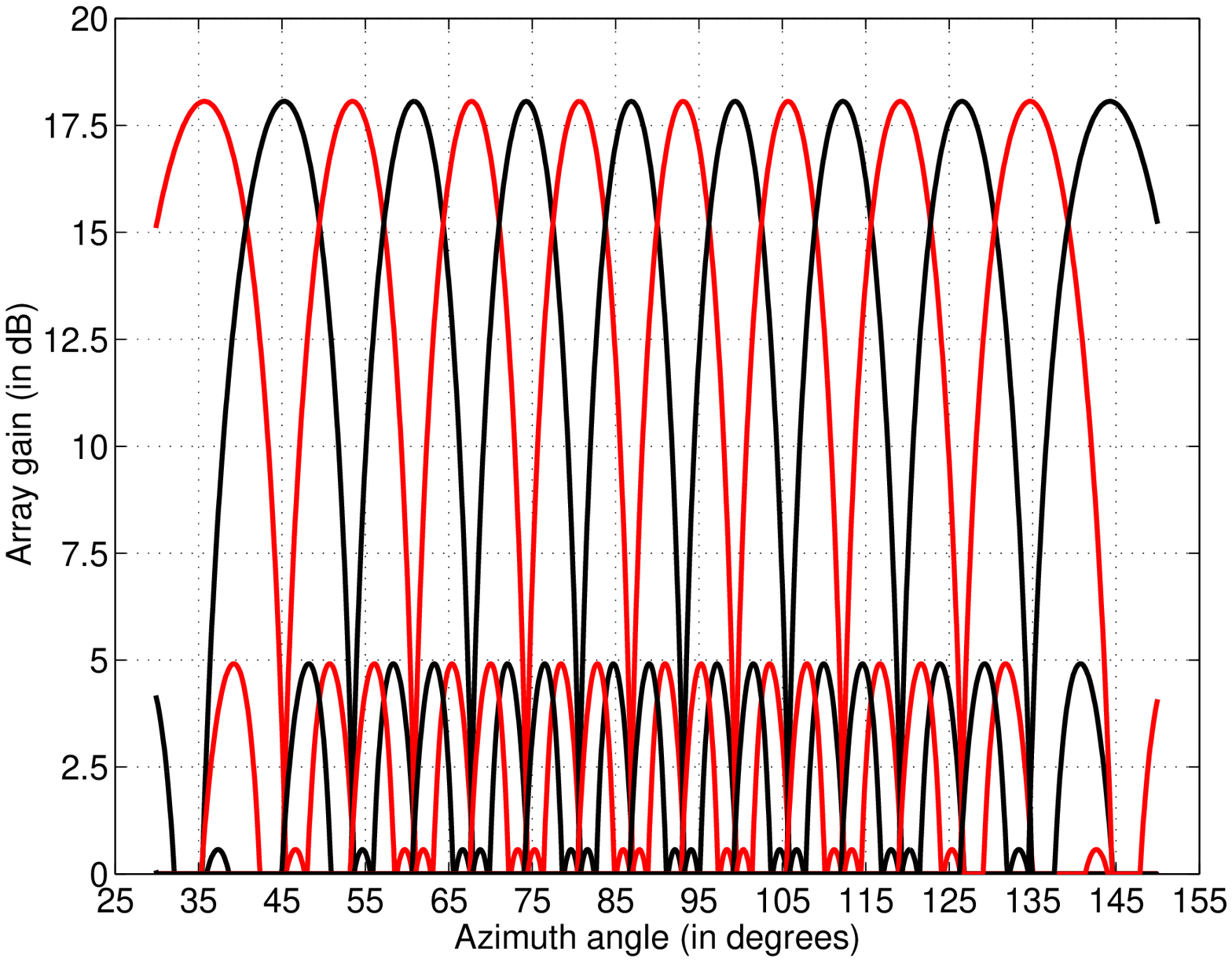}
\\
(a) & (b)
\\
\includegraphics[height=2.5in,width=3.1in] {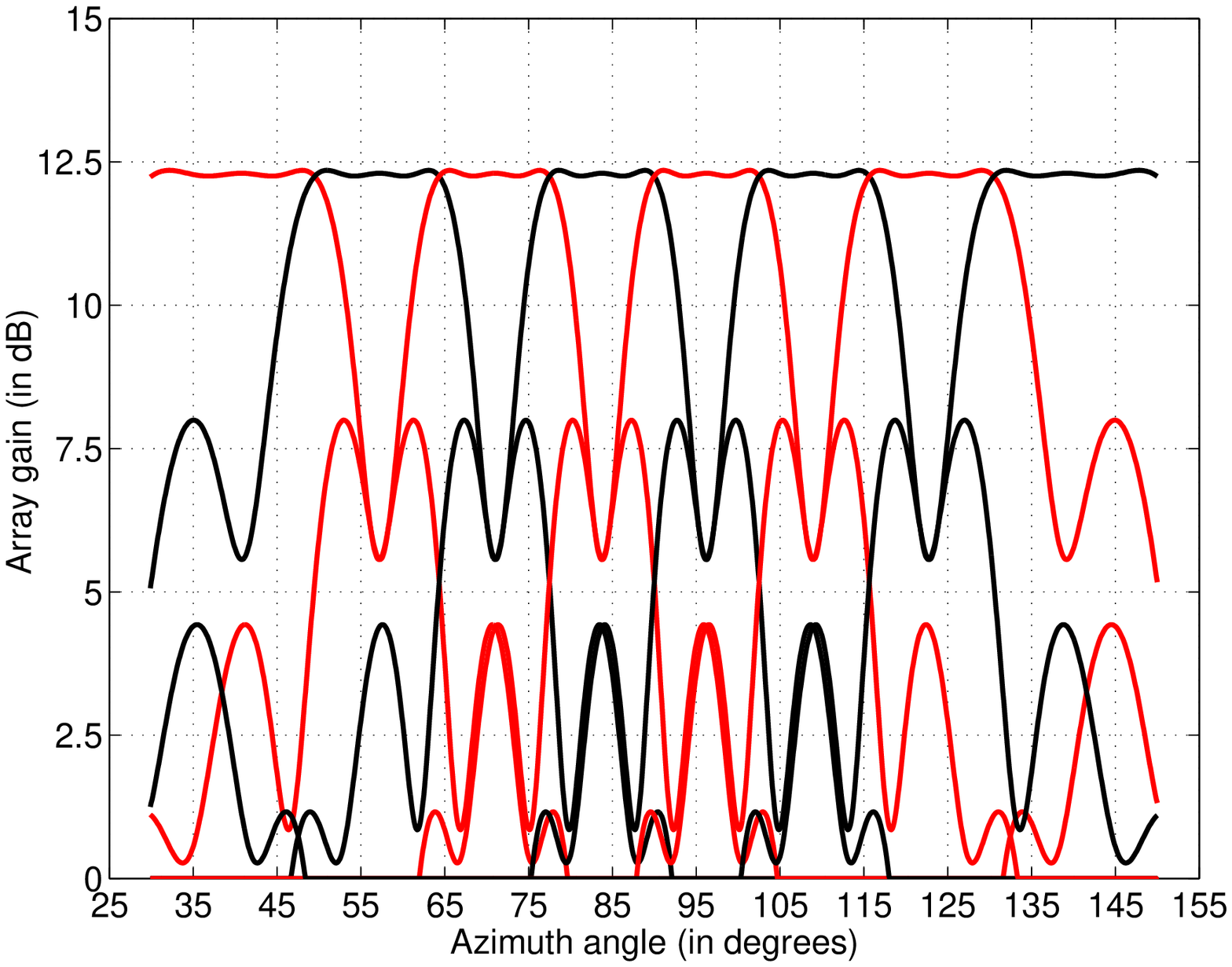}
&
\includegraphics[height=2.5in,width=3.1in] {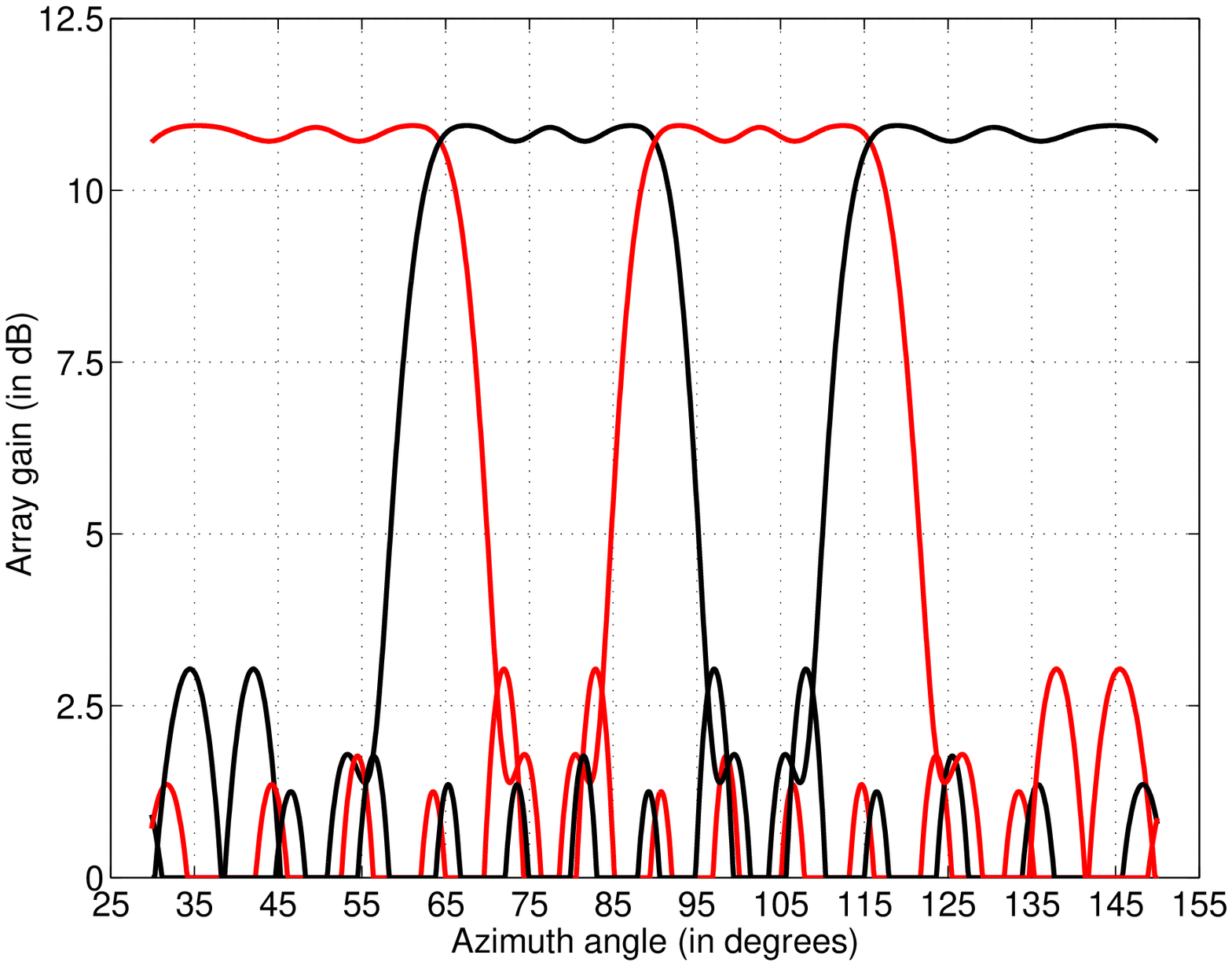}
\\
%(c) & (d)
(a) & (b)
\end{tabular}
\caption{\label{fig_beampatterns_codebooks}
Beam patterns in the azimuth plane of four different base-station codebooks, all
covering a $120^{\sf o} \times 30^{\sf o}$ coverage area, with (a) $N = 32$,
(b) $N = 16$, (c) $N = 8$, and (d) $N = 4$ elements in ${\cal F}_{\sf tr}$.}
%(a) $N = 8$ and (b) $N = 4$ elements in ${\cal F}_{\sf tr}$.}
\end{center}
%\vspace{-5mm}
\end{figure*}

At this stage, it is worth noting that a number of system parameters impact the
performance of the proposed multi-user schemes such as: i) Granularity of
${\cal F}_{\sf tr}$ and ${\cal G}_{\sf tr}^k$ (initial beam alignment codebook
sizes), ii) Coarseness of channel approximation (rank-$P$), iii) Finite-rate
feedback of channel reconstruction parameters, and iv) Quantization of the resulting
multi-user beams.

\begin{figure*}[htb!]
\begin{center}
\begin{tabular}{cc}
\includegraphics[height=2.5in,width=3.0in]{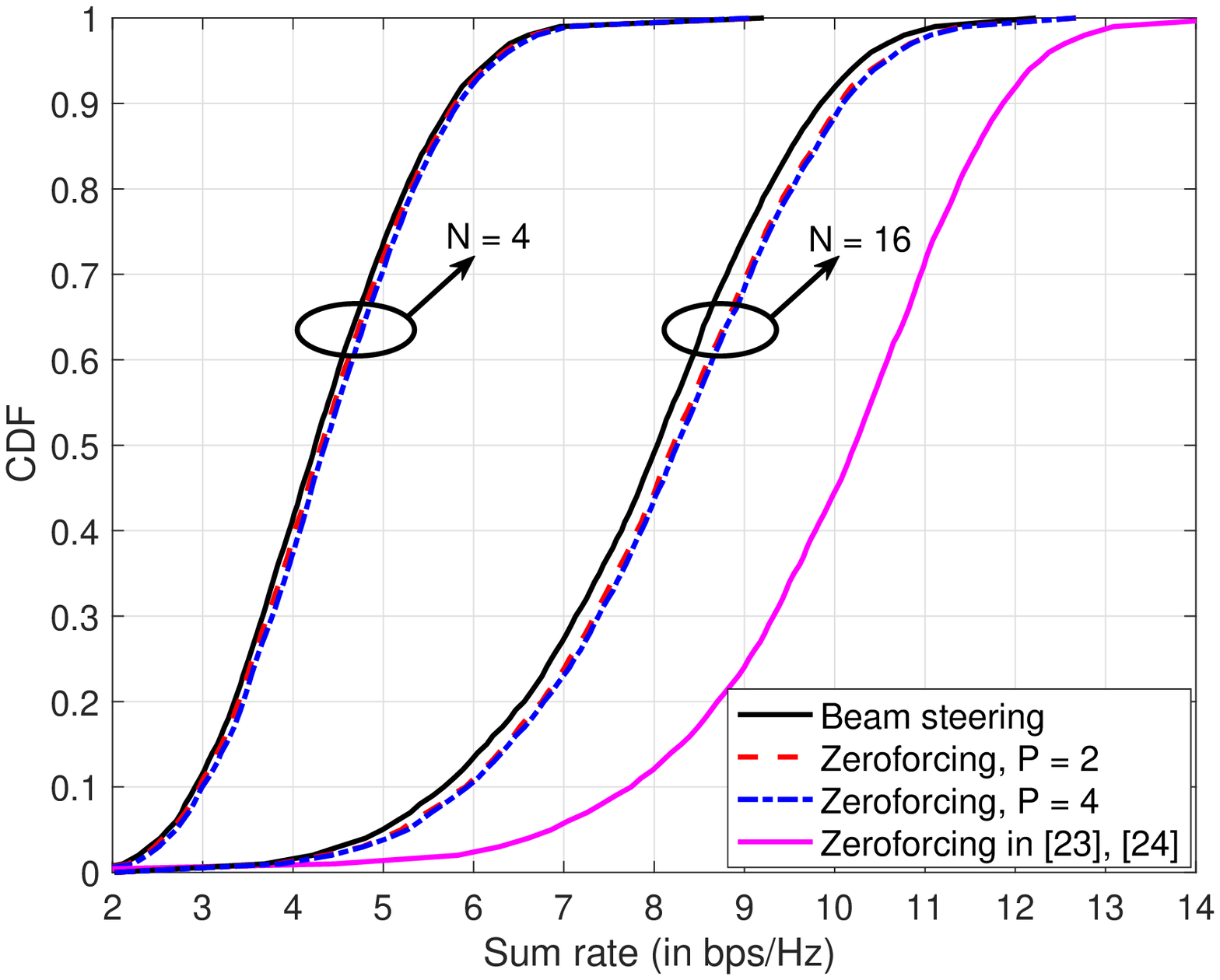}
& {\hspace{0.1in}}
\includegraphics[height=2.5in,width=3.0in]{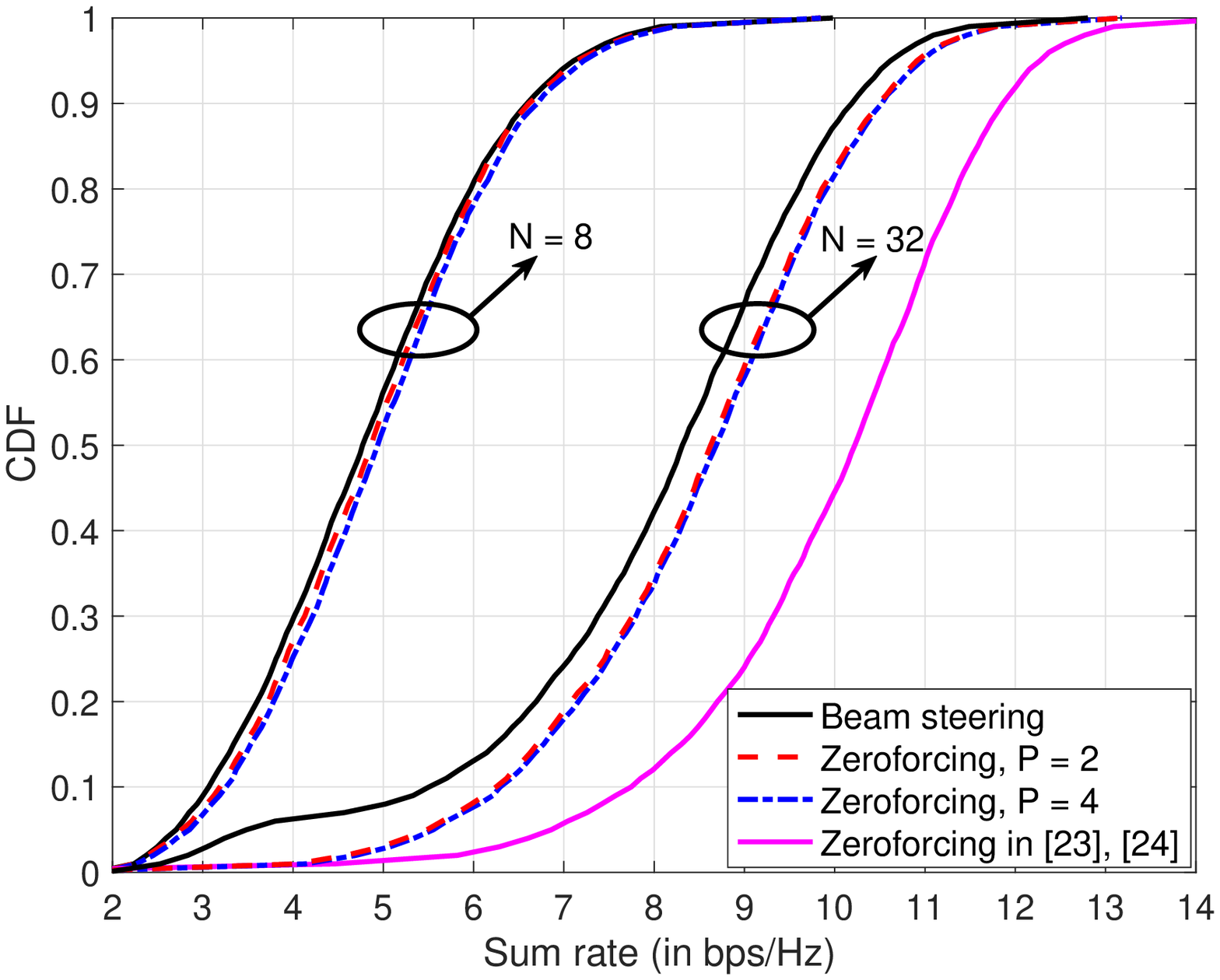}
\\
{\vspace{0.15in}}
(a) & (b) \\
\includegraphics[height=2.5in,width=3.0in]{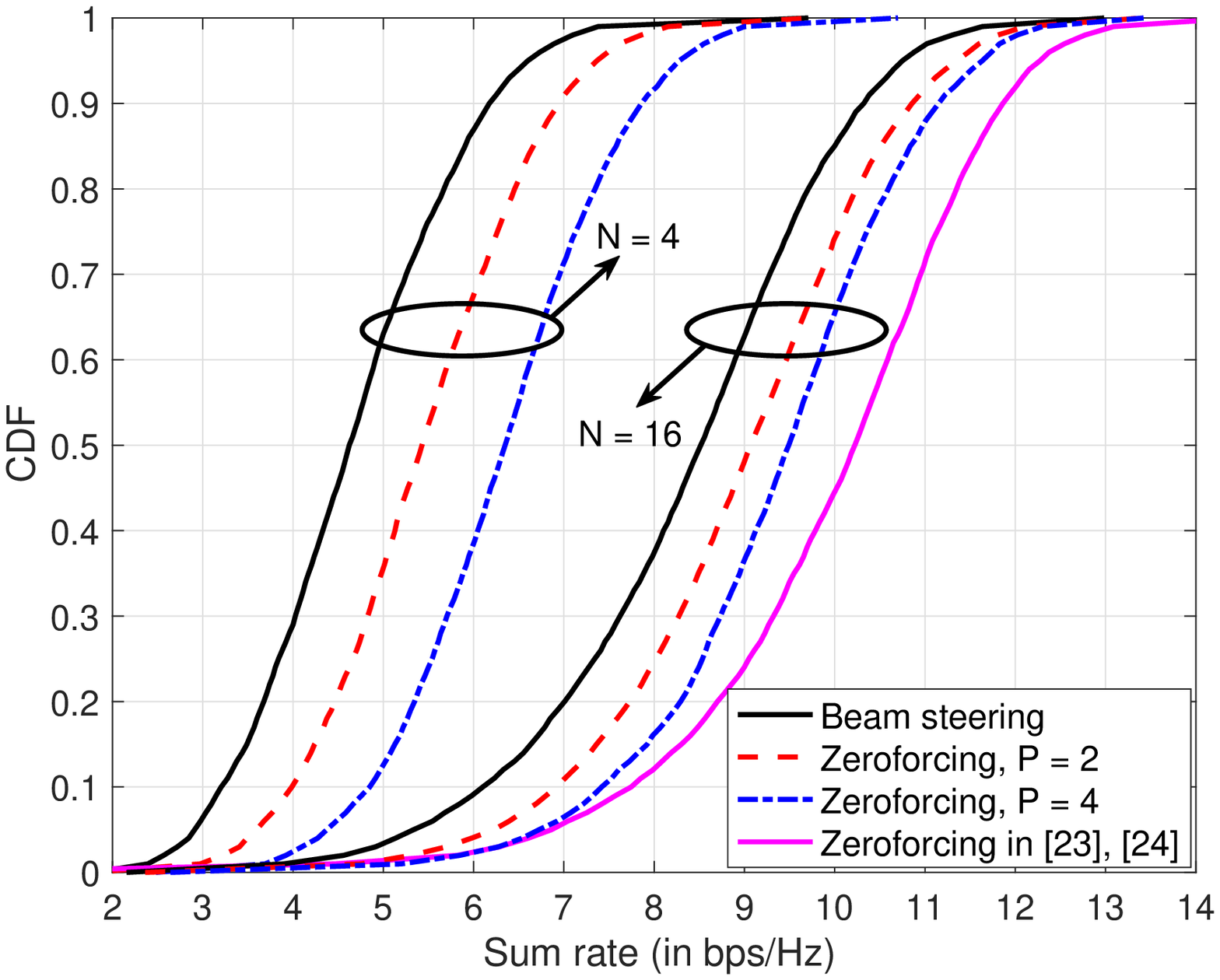}
& {\hspace{0.1in}}
\includegraphics[height=2.5in,width=3.0in]{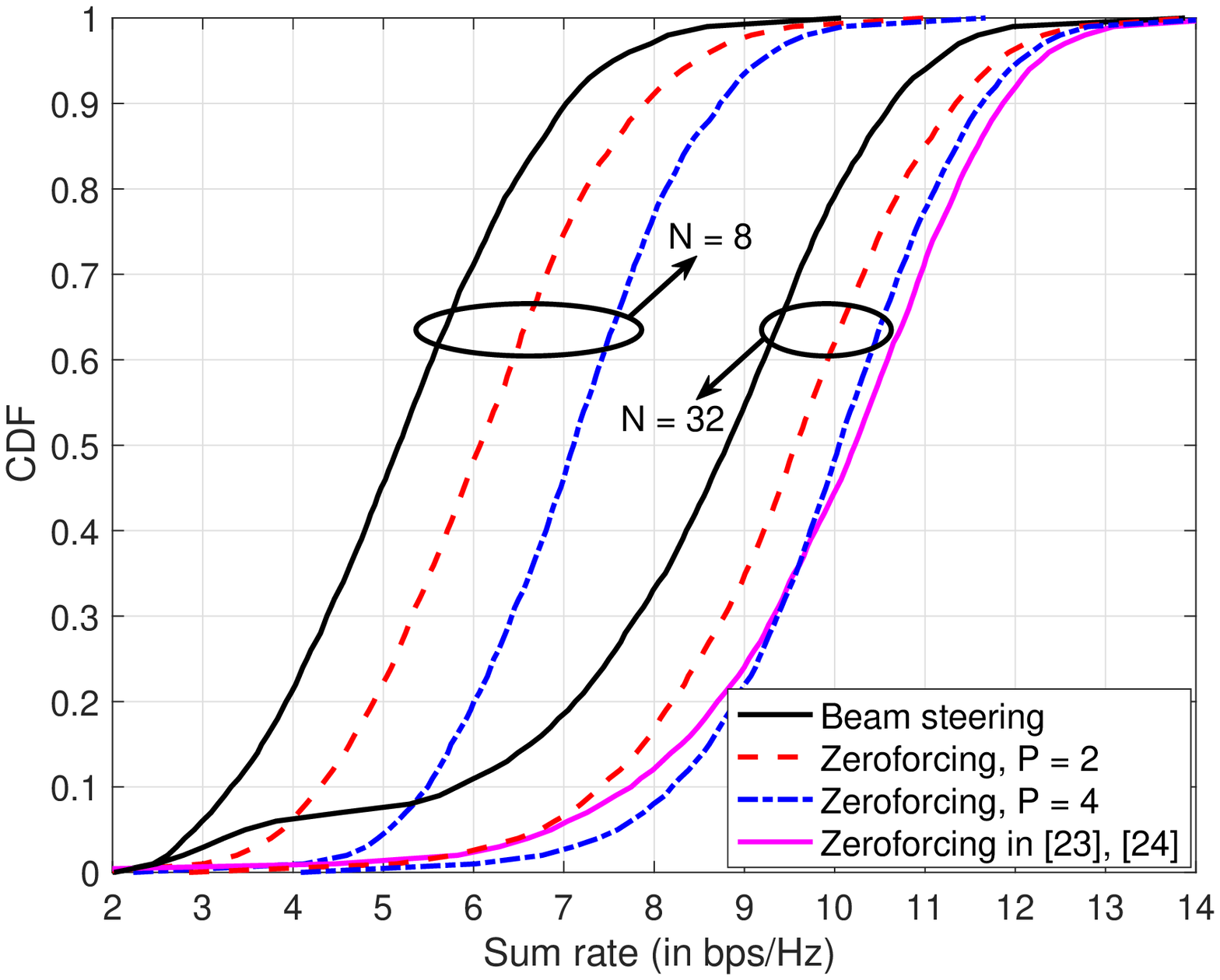}
\\
(c) & (d)
\end{tabular}
\caption{\label{fig:txCodebook}
CDF of sum rates for a beam steering scheme and the proposed zeroforcing scheme for
different choices of $N$ with $M = 4$ in (a) and (b), and $M = 16$ in (c) and (d).}
\end{center}
% \vspace{-5mm}
\end{figure*}

\subsection{Impact of Initial Beam Alignment Codebook}
\label{sec5a}
In the first study, we consider the relative performance of the zeroforcing
scheme (proposed in Prop.~\ref{prop_zf}) relative to a baseline beam steering
scheme with different initial beam alignment codebooks. We assume that the
system has infinite-precision feedback of channel reconstruction parameters
and infinite-precision resolution in the quantization of multi-user beams. We
also compare the performance of the proposed schemes with the zeroforcing scheme
presented in~\cite{vasanth_gcom16,vasanth_jsac2017}, where the system is assumed
to be able to find perfectly aligned directional beams in the training phase.
%Henceforth, the zeroforcing scheme outlined in \cite{vasanth_gcom16, vasanth_jsac2017}
%will be referred to as the $\sf MU-ZF$ scheme.
Fig.~\ref{fig:txCodebook} illustrates this comparative performance with different
choices of $P$ in approximating ${\bf g}_k^{\dagger} \widehat{\bf H}_k$ and different
codebook sizes ($N$ and $M$).
%corresponding to the codebooks illustrated in Fig.~\ref{fig_beampatterns_codebooks}.

While it is intuitive that there should be diminishing performance as $P$ increases
(since increasing $P$ beyond the channel rank $L_k$ is not expected to improve
performance), whether this saturation in
performance is observed with a low-rank channel approximation is dependent on the
resolution of the codebooks. In particular, increasing $P$ %improving the channel approximation
when the codebook granularity is already poor (small $M$ and $N$) does not lead to
any performance improvement than observed with $P = 1$ (beam steering). On the other hand,
with a high resolution for ${\cal F}_{\sf tr}$ (large $N$), even a rank-$2$
approximation appears to be sufficient to reap most of the performance improvement
gains. This is because the performance of the baseline (beam steering) scheme is
already quite good and significant relative improvement over it with increasing $P$
has a lower likelihood unless the channel has a large number of similar gain clusters (a
low-probability event). When $M$ is large and $N$ is small, the beam steering performance
is poor and the channel can be better approximated with the higher codebook resolution
of ${\cal G}_{\sf tr}^k$ leading to a sustained performance improvement for even
up to $P = 4$. For example, with $N = 4$ or $8$ and $M = 16$, zeroforcing based on a
rank-$4$ channel approximation leads to around $2$ bps/Hz improvement at the median level.

In terms of performance comparison, note that the scheme
from~\cite{vasanth_gcom16,vasanth_jsac2017} assumes $P = 1$ but infinite-precision
in terms of beam alignment ($N = M \rightarrow \infty$). Thus, it is not surprising
that as $N$ and $M$ increase, the performance of the proposed schemes compare well
with that of~\cite{vasanth_gcom16,vasanth_jsac2017}. For lower codebook
resolutions, the proposed schemes overcome the codebook disadvantage by leveraging
a better channel approximation as $P$ increases.
These observations suggest that the optimal choice of the rank in approximating
${\bf g}_k^{\dagger} \widehat{\bf H}_k$ (which in turn determines the feedback overhead)
depends not only on the rank of the true channel ${\bf H}_k$, but also on the codebook
granularities. In general, a higher $P$ (and feedback overhead) is necessary if the
codebook resolution is rich enough at the user end to allow the parsing of the channel
better, but poor enough at the base-station end to allow a sustained performance
improvement with increasing $P$. In particular, we provide the following heuristic
design guidelines based on our studies
\begin{eqnarray}
P = \left\{ \begin{array}{cl}
1 & {\sf if} \hspp M \hspp {\sf and} \hspp N \hspp {\sf are} \hspp {\sf small} \\
2 & {\sf if} \hspp M \hspp {\sf is} \hspp {\sf small} \hspp {\sf and} \hspp
N \hspp {\sf is} \hspp {\sf large}
\\
4 & {\sf if} \hspp M \hspp {\sf is} \hspp {\sf large}.
\end{array}
\right.
\end{eqnarray}

\begin{figure*}[htb!]
\begin{center}
\begin{tabular}{cc}
\includegraphics[height=2.5in,width=3.0in] {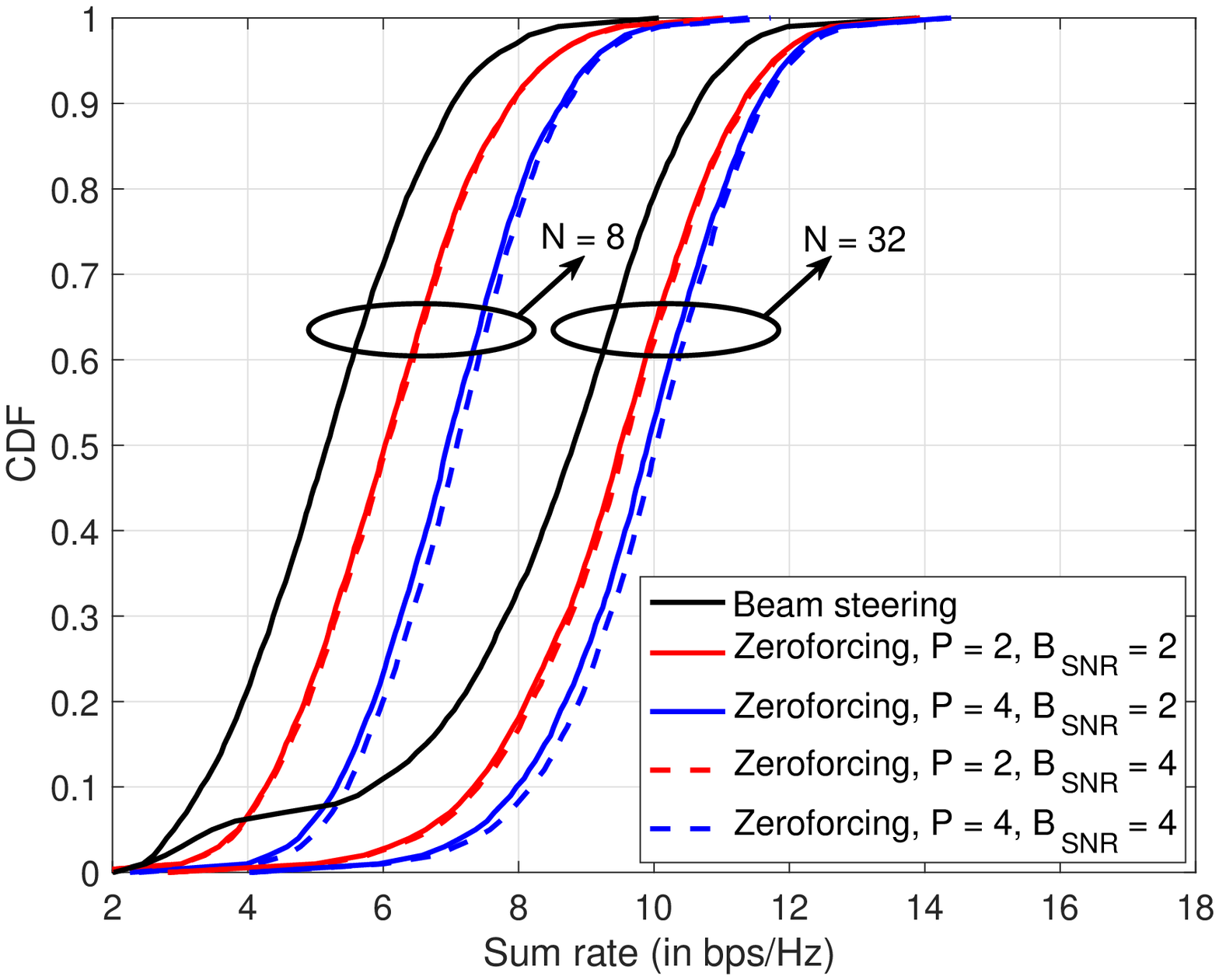}
& {\hspace{0.1in}}
  \includegraphics[height=2.5in,width=3.0in] {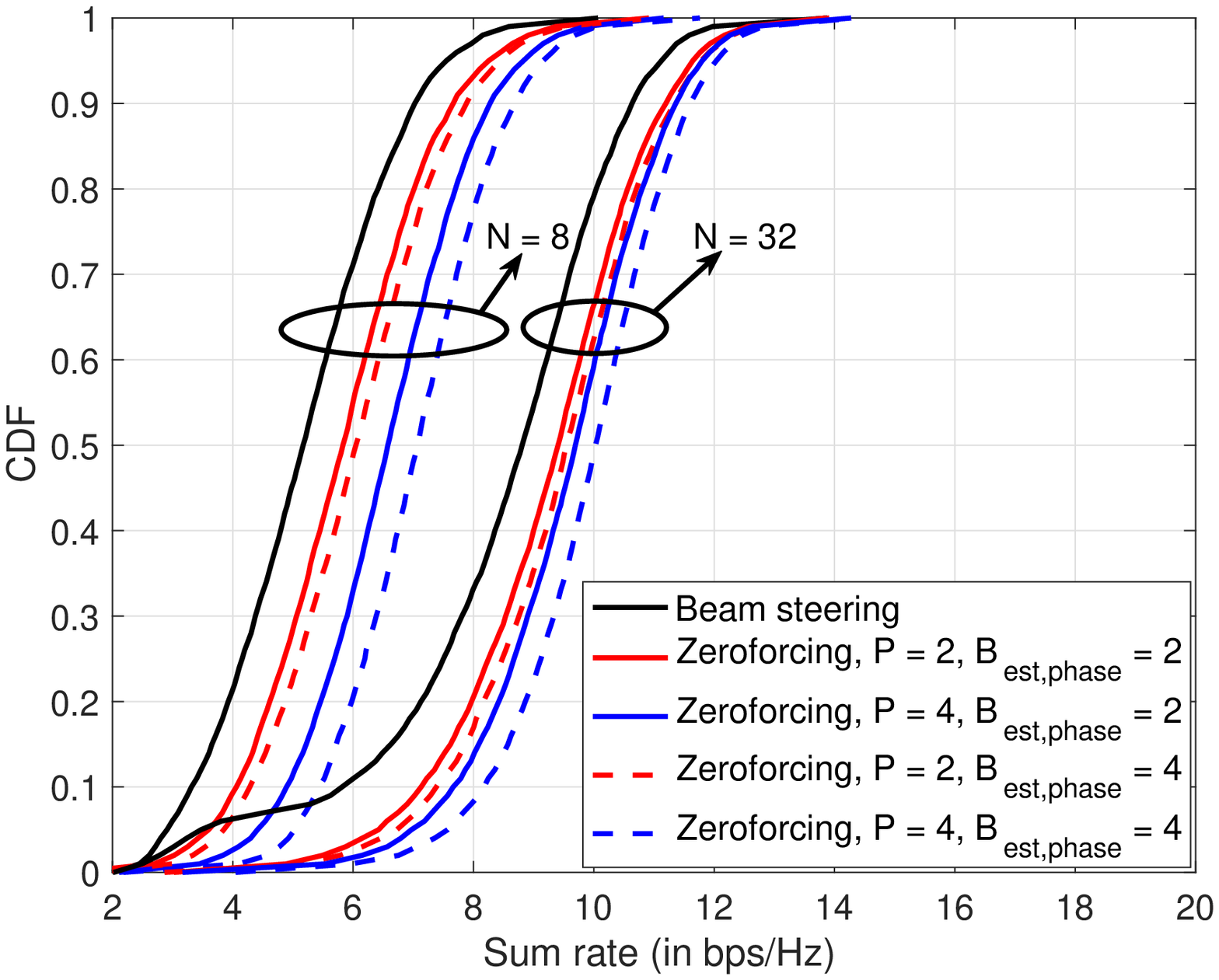}
\\
(a) & (b)
\\
  \includegraphics[height=2.5in,width=3.0in]{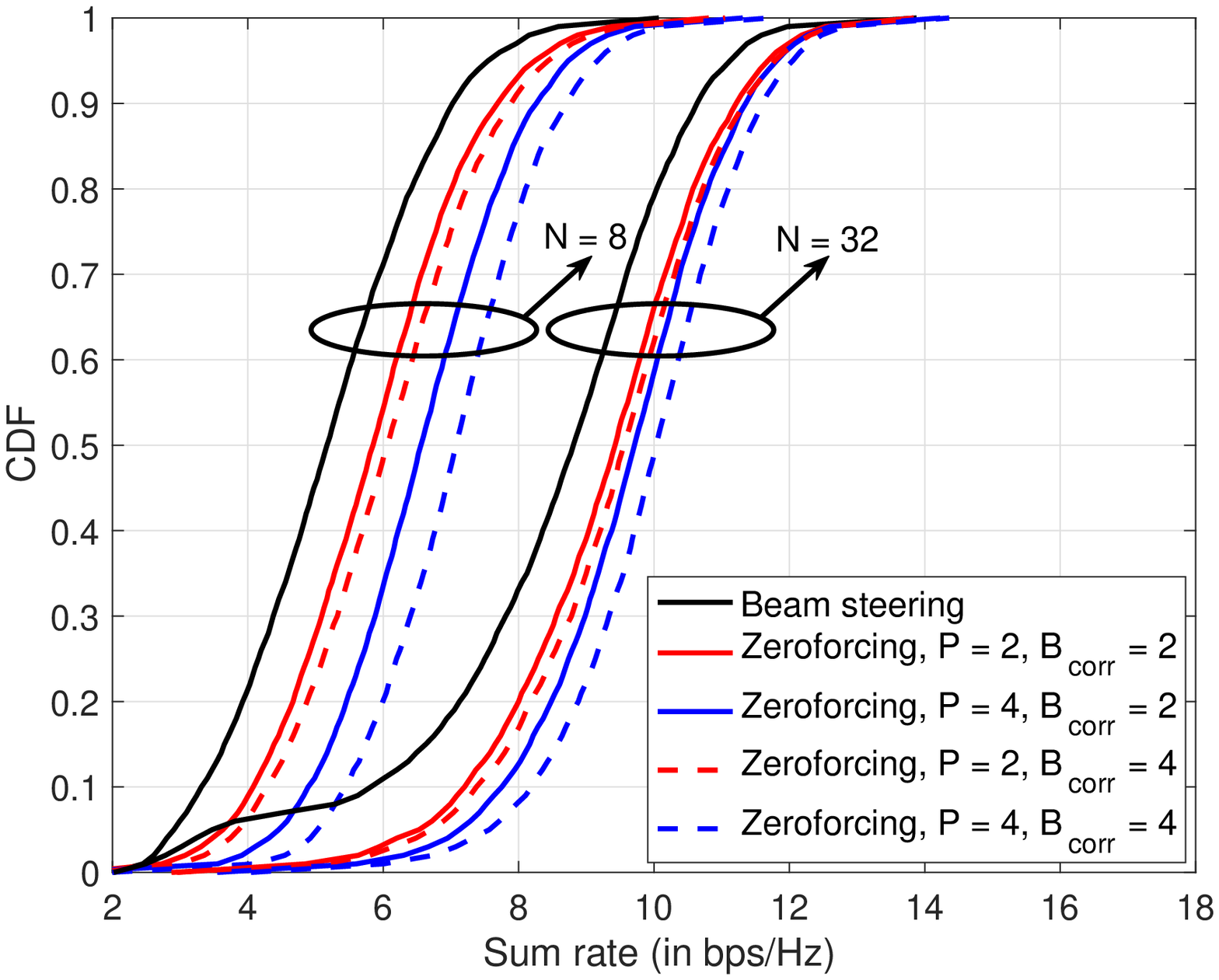}
& {\hspace{0.1in}}
  \includegraphics[height=2.5in,width=3.0in]{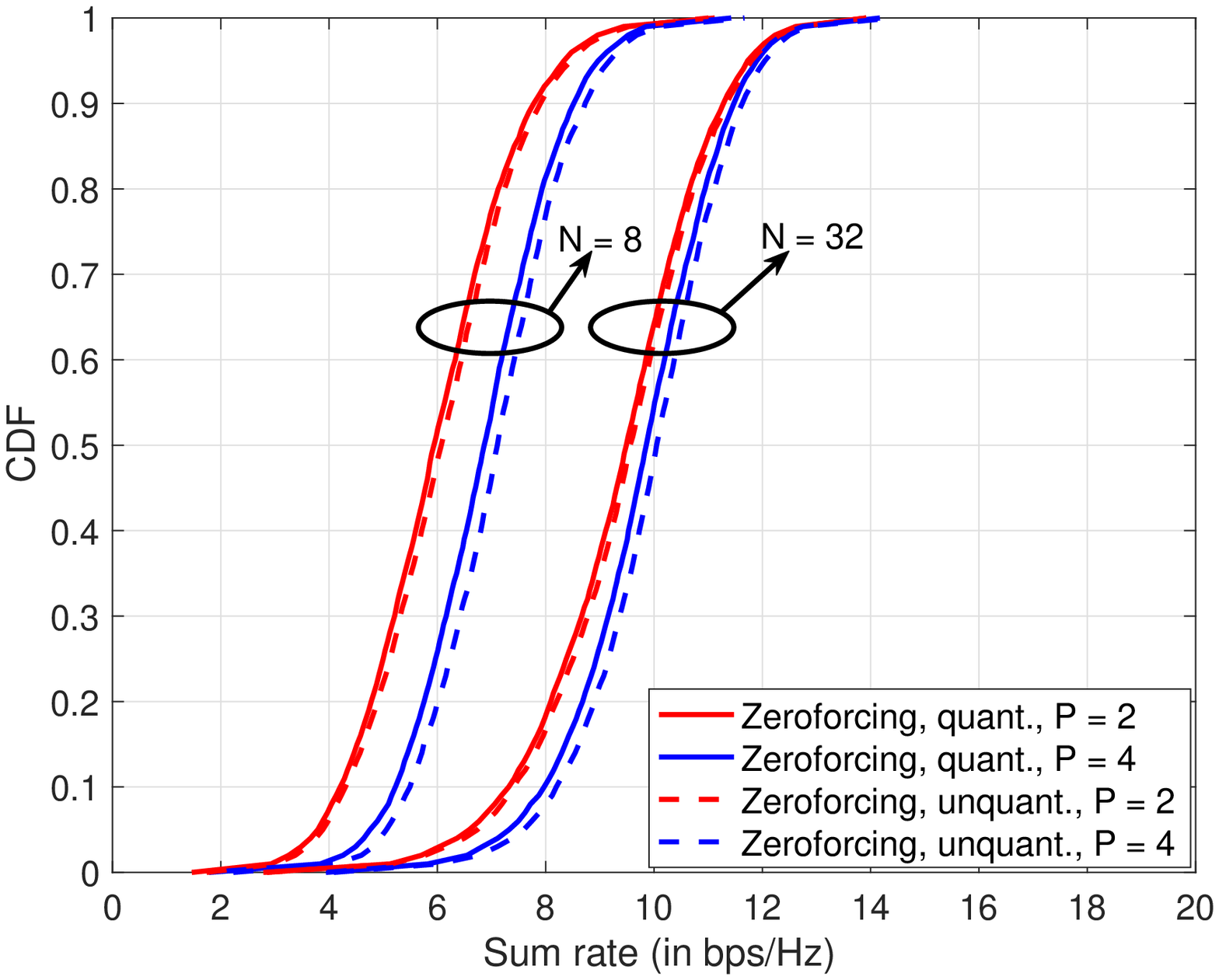}
\\
(c) & (d)
\end{tabular}
\caption{\label{fig:fdbkQuant}
CDF of sum rates of the different multi-user schemes with finite-rate feedback
of (a) only received ${\sf SNR}$s, (b) only received signal phases,
(c) only user side cross-correlation information, and (d) all the parameters
quantized simultaneously.
%In (d), 3 bits are allocated to each quantized feedback quantity.
}
\end{center}
%\vspace{-5mm}
\end{figure*}

\subsection{Quantizer Design}
\label{sec3c}
Towards the second study, we utilize different quantization functions to
quantize the different parameters needed in channel reconstruction. For a
phase term $\theta$ with a dynamic range of $[0, 2\pi)$
(e.g., $\angle{ \widehat{\bf s}_{ {\sf tr}, k,\ell} }$ and $\angle{ \beta_{k,\ell} }$), we use a
uniform quantizer of the form
\begin{eqnarray}
{\cal Q}_{B}(\theta) = \frac{2 \pi}{2^B} \cdot {\sf round}\left( \frac{2^B}{2 \pi}
\cdot \theta %\hspp {\sf mod} \hspp 2 \pi)
\right),
\label{eq_quant_phase}
\end{eqnarray}
where ${\sf round}(\cdot)$ stands for a function that rounds off the underlying quantity to
the nearest integer. For an amplitude term $\alpha$ with a dynamic range of $[0, 1]$
(e.g., $|\beta_{k,\ell}|$), we use a non-uniform quantizer of the form
\begin{eqnarray}
{\cal Q}_B(\alpha) = \frac{ {\sf round}\left( (2^B - 1) \cdot \alpha \right) }{2^B - 1}.
\label{eq_ampl_quantize}
\end{eqnarray}
The reason for scaling with respect to $2^B - 1$ in~(\ref{eq_ampl_quantize}) instead of by
$2^B$ is because we want the quantized set to include both $0$ and $1$ for proper
cross-correlation quantization. For example, in the typical case where the multi-user
reception beam ${\bf g}_k = {\bf g}_{{\sf tr}, m_1^{k}}^{(k)}$, we have $|\beta_{k,1}| = 1$
and the use of a uniform amplitude quantizer will not allow the correct reproduction of
this important quantity at the base-station end.

Quantization of the ${\sf SNR}$ is performed on a dB scale rather than on a linear
scale. This is intuitive since ${\sf SNR}$ measurements have a wide dynamic range.
The proposed ${\sf SNR}$ quantizer is similar to quantizations considered in Fourth 
Generation (4G) systems. In particular, for a received ${\sf SNR}$ term 
$\varrho$ (in dB) with a
theoretically unbounded range (e.g., $10 \log_{10} \big( {\sf SNR}_{ {\sf rx},
\hsppp \ell }^{(k)} \big)$), we first cap $\varrho$ to a maximum value of
$\varrho_{\sf max}$ and quantize a spread of $\Delta$ (in dB) with $2^B$ quantization
levels (denoted as $\varrho_i$) as follows:
\begin{eqnarray}
\varrho_i = \varrho_{\sf max} - \frac{\Delta}{2^B - 1} \cdot i , \hspp i = 0, \cdots,
2^B - 1.
\end{eqnarray}
The quantization of $\varrho$ is given as
\begin{eqnarray}
{\cal Q}_B(\varrho) = \varrho_{i^{\star}} \hspp {\sf where} \hspp
i^{\star} = \arg \min \limits_{i = 0, \cdots, 2^B - 1} |\varrho - \varrho_i | .
\end{eqnarray}
The parameters $\varrho_{\sf max}$ and $\Delta$ correspond to the maximum quantizer
level value and the distance between adjacent quantizer levels, respectively. In 
our numerical
studies, we use $\varrho_{\sf max} = 30$ dB with $\Delta = 24$ dB for $B = 2$ bits, 
and $\Delta = 30$ dB for $B = 4$ bits.

A similar approach is pursued in quantizing the amplitudes of the multi-user beam.
While these amplitudes do not span a wide range, the relative variation across the
antenna array can show wide variations. Specifically,
the infinite-precision zeroforcing beams generated in Prop.~\ref{prop_zf} are
quantized to meet the RF constraints in~(\ref{eq_hybrid_precoder_constraints})
as described next. Since $\| {\bf f}_k \| = 1$, we assume that on average
${\bf f}_k(i) \approx \frac{1}{ \sqrt{ N_{\sf t} } }$. By scaling $|{\bf f}_k(i)|^2$
by $N_{\sf t}$, we can ensure that $10 \log_{10} \left(N_{\sf t} \cdot
|{\bf f}_k(i)|^2 \right)$ is centered around $0$ dB and for this quantity, we
generate $2^B$ quantization levels in dB scale (denoted as $f_i$) corresponding to
a step size of $\Delta f$ (in dB) as follows:
\begin{eqnarray}
f_i = \Delta f \cdot \left[ i + 1 - 2^{B-1} \right], \hspp i = 0, \cdots, 2^B -1.
\end{eqnarray}
With these levels that are spaced $\Delta f$ apart, we obtain the quantized beam weights as
\begin{eqnarray}
| {\widehat{\bf f}}_k(i) | = \widetilde{\cal Q}_B( | {\bf f}_k(i) | )
= \frac{1}{ \sqrt{ N_{\sf t}} } \cdot \left\{
\begin{array}{cl}
0 & {\sf if} \hspp 10 \log_{10} \left(N_{\sf t} \cdot |{\bf f}_k(i)|^2 \right)
< - \Delta f \cdot (2^{B-1} - 1) \\
10^{ \frac{ f_{j^{\star}} }{20} } & {\sf otherwise},
\end{array}
\right.
\end{eqnarray}
where
\begin{eqnarray}
j^{\star} = \arg \min_{j }
10 \log_{10} \left(N_{\sf t} \cdot |{\bf f}_k(i)|^2 \right) - f_j
\hspp {\sf provided} \hspp 10 \log_{10} \left(N_{\sf t} \cdot |{\bf f}_k(i)|^2 \right)
> f_j.
\label{eq_constraint}
\end{eqnarray}
The constraint in~(\ref{eq_constraint}) ensures that $\sum_i | {\widehat{\bf f}}_k(i) |^2
\leq 1$. In our numerical studies, we use $\Delta f = 1$ dB for $B = 4$ bits leading
to a range of $-7$ to $8$ dB for $f_i$. We also use $\Delta f = 0.25$ dB for $B = 6$
bits leading to a range of $-7.75$ to $8$ dB for $f_i$. For the phase quantities (that
is, $\widetilde{\cal Q}_B( \angle{ {\bf f}_k(i) } )$), we reuse
${\cal Q}_B( \angle{ {\bf f}_k(i) } )$ as in~(\ref{eq_quant_phase}).

\subsection{Finite-Rate Feedback}
\label{sec5b}
With the quantizer design as described in~Sec.~\ref{sec3c},
%In the second study,
we now consider the impact of finite-rate feedback of the
quantities of interest necessary for the channel reconstruction step. As noted
from Table~\ref{table_feedback}, each user quantizes and feeds back to the
base-station: i) the base-station beam indices, ii) the received
${\sf SNR}$s, iii) the received signal's phases, and iv) user side codebook
correlation information (amplitude and phases). To reduce clutter in
presentation, in our studies illustrated in Fig.~\ref{fig:fdbkQuant}, we
only focus on the $N = 8$ and $N = 32$ codebooks for beam alignment with
$M = 16$ at the user side. Fig.~\ref{fig:fdbkQuant}(a) considers the impact of
$B_{\sf SNR}$ (the number of bits used in received ${\sf SNR}$ quantization) while
infinite-precision is used for the signal phases and codebook correlation. This
figure shows that the proposed scheme is robust to $B_{\sf SNR}$ in the sense
that for both the $P = 2$ and $P = 4$ cases, the performance improvement is minimal
as $B_{\sf SNR}$ is increased from $2$ bits to $4$ bits.

On the other hand, Fig.~\ref{fig:fdbkQuant}(b) considers the impact of
$B_{ {\sf est}, \hsppp {\sf phase}}$ (the number of bits used in received
signal phase quantization) while infinite-precision is used for received ${\sf SNR}$
and codebook correlation. In the third experiment, we study the impact
of codebook correlation quantization with infinite-precision for the other two quantities.
To simplify this investigation, we assume that $B_{\sf corr, \hsppp phase} =
B_{\sf corr, \hsppp amp} = B_{\sf corr}$ and Fig.~\ref{fig:fdbkQuant}(c) considers
the impact of $B_{\sf corr}$ on performance. Both Figs.~\ref{fig:fdbkQuant}(b)
and~(c) show that increasing $B_{ {\sf est}, \hsppp {\sf phase}}$ or $B_{\sf corr}$
has maximal impact on performance for %small $N$ or
large $P$. In other words, if
%the baseline performance is poor because of poor codebook resolution or if
the channel approximation gets better, it becomes pertinent to quantize the
phase terms and codebook correlation information in the channel reconstruction
with a finer resolution.

While Figs.~\ref{fig:fdbkQuant}(a)-(c) study the quantization of each parameter
of interest separately, we now consider the impact of finite-rate quantization of
all the parameters necessary for channel reconstruction (relative to infinite-precision
quantization). For this, we consider the case where $B_{\sf SNR} =
B_{ {\sf est}, \hsppp {\sf phase}} = B_{ {\sf corr}, \hsppp {\sf amp}} =
B_{ {\sf corr}, \hsppp {\sf phase}} = 3$ bits with $M = 16$. From Fig.~\ref{fig:fdbkQuant}(d),
we observe that the proposed joint quantization scheme performs comparable with a scheme that
uses infinite-precision for all the parameters of interest.

\begin{table*}[htb!]
\caption{Feedback overhead ($B_{\sf feedback}$) for different choices of $P$ and $N$}
\label{table_fb_overhead}
\begin{center}
\begin{tabular}{|c||c|c|c|c|}
\hline
& $N = 4$ & $N = 8$ & $N = 16$ & $N = 32$ \\
\hline
$P = 2$ & $%B_{\sf feedback} =
14$ & $%B_{\sf feedback} =
16$
& $%B_{\sf feedback} =
18$ & $%B_{\sf feedback} =
20$ \\
\hline
$P = 4$ & $%B_{\sf feedback} =
44$ & $%B_{\sf feedback} =
48$
& $%B_{\sf feedback} =
52$ & $%B_{\sf feedback} =
56$
\\
\hline
%\hline
\end{tabular}
\end{center}
\end{table*}
%\end{center}

At this stage, it is important to note that the feedback overhead of $\varphi_{k,\ell}$ and
$\nu_{k,\ell}$ can be combined\footnote{Similarly, it might be envisioned that the feedback
of $\gamma_{k,\ell}$ and $\mu_{k,\ell}$ can be combined, but their dynamic ranges are different.
Feedback overhead reduction could be a useful topic of study in future research.} since
they are always used in the form $\varphi_{k,\ell} + \nu_{k,\ell}$
(see~(\ref{eq_eigenspace_reconst})). Thus, based on the above studies, we make the following
heuristic design guidelines on the feedback overhead
\begin{eqnarray}
B_{\sf SNR}  %&= &
= 2 \hspp {\sf bits}, \hspp \hspp
B_{ {\sf est}, \hsppp {\sf phase}} + B_{ {\sf corr}, \hsppp {\sf phase}} =
B_{{\sf corr}, \hsppp {\sf amp}} = %& = &
P \hspp {\sf bits}.
%\left\{ \begin{array}{cl}
%2 & {\sf if} \hspp P = 2 \\
%4 & {\sf if} \hspp P = 4.
%\end{array} \right.
\end{eqnarray}
Combining this information with Table~\ref{table_feedback}, the total feedback
overhead from each user is given as %to the base-station is
\begin{align}
B_{\sf feedback} & = P \cdot \left[ \log_2(N) + B_{\sf SNR} + B_{ {\sf corr}, \hsppp {\sf amp}}
\right] + (P-1) \cdot \left[ B_{ {\sf est}, \hsppp {\sf phase}} +
B_{ {\sf corr}, \hsppp {\sf phase}}  \right] \hspp \hspp ({\sf in} \hspp {\sf bits})
\\
& = P \cdot \left[ \log_2(N) + 2 + P \right] + (P-1) \cdot P
\\
& = P \cdot \left[ \log_2(N) + 2P + 1 \right] \hspp {\sf bits}.
\end{align}
$B_{\sf feedback}$ is presented in Table~\ref{table_fb_overhead} for the choices
$P \in \{2, 4\}$ and $N \in \{ 4, 8, 16 , 32 \}$. From Table~\ref{table_fb_overhead},
a $56$ bit control payload appears to be sufficient to convey the information
necessary for multi-user beamforming across different choices of $M$, $N$ and $P$.
On a first glance, while this may appear to be onerous, similar feedback overheads
are currently considered viable in 3GPP 5G-NR design. In particular, two types
of feedback methods are being
studied~\cite[Sec.\ 8.2.1.6.3, pp.\ 24-26]{3gpp_CM_rel14_38912}: i) Type-I feedback
of both the beam indices and RSRPs of the top-$4$ beams, and ii) a more general
Type-II feedback that can include feedback of covariance matrices, co-phasing
factors with different codebook structures, etc. Further, the time-scales at
which this information has to be reported is on the order of the coherence time
of the channel (which varies from a few milliseconds at high speeds to a few
%hundreds of milliseconds in an indoor scenario~\cite{vasanth_comm_mag_16,vutha_va})
hundreds of milliseconds in an indoor or low speed 
scenario~\cite{vasanth_comm_mag_16,vasanth_tap_blockage})
allowing multiple long PUCCH %physical uplink control channel (PUCCH)
instances for beam reporting. Also, this control information can be fed back 
on legacy carriers such as 4G links in a non-standalone deployment. Thus, 
the feedback
overhead necessary for realizing the proposed schemes are practically viable.

\subsection{Quantization of Multi-User Beams and Comparison with Upper Bounds}
\label{sec5c}
In the third study, the effect of quantizing the multi-user beams to ensure that
it fits the RF precoder constraints as in~(\ref{eq_hybrid_precoder_constraints})
is considered. In general, if a low rate quantization is used ($B_{\sf amp}$ or
$B_{\sf phase}$) as $P$ increases, the resultant multi-user beam's sum rate
performance could be worse than that with beam steering. In particular,
from Fig.~\ref{fig:bfQuant}(a), we observe that a higher phase resolution
($B_{\sf phase}$) is necessary for improved performance as the codebook resolution
improves (large $N$) or when $P$ increases. On the other hand, from
Fig.~\ref{fig:bfQuant}(b), we observe that an amplitude resolution $B_{\sf amp}$)
on the order of $4$-$6$ bits can produce a performance comparable with the unquantized
scheme.

In Fig.~\ref{fig:upperBounds}, we finally compare the performance of the proposed
zeroforcing scheme with the beam steering scheme and the bounds established in
Sec.~\ref{sec4}. We also benchmark the performance with a fully-digital system
employing: i) maximal ratio transmission/maximal ratio combining (MRT/MRC) beams
in the initial alignment phase, and ii) a zeroforcing scheme performed using the
MRT/MRC beams as in~\cite{vasanth_gcom16,vasanth_jsac2017}. Note that the MRT/MRC 
scheme is different from that employed in~\cite{vasanth_gcom16,vasanth_jsac2017} 
where perfect beam steering vectors are used in deriving the zeroforcing 
structure. In terms of differences between these structures, the readers are 
referred to~\cite{vasanth_gcom15}. For the proposed scheme, an $M = 16$ codebook 
is used at the user end. Figs.~\ref{fig:upperBounds}(a)-(b) and Fig.~\ref{fig:upperBounds1} 
illustrate the trends with $N = 8$, $N = 32$, and $N = 256$ codebooks, respectively. 
For $N = 256$, we employ a DFT codebook at the base-station covering the $120^\circ 
\times 30^\circ$ AoD space.

With low-resolution quantization, %($N = 8$),
we note that there
is a considerable performance gap between the zeroforcing scheme and the scalar
optimization-based upper bound (up to $2$ bps/Hz). On the other hand, this gap %considerably
reduces as $N$ increases suggesting the good
performance of the zeroforcing scheme. Nevertheless, the performance gap between the
proposed zeroforcing scheme and the upper bounds suggests the possible utility of
more advanced feedback mechanisms, a topic for future research. In all the plots, there
is a considerable gap between the performance of the upper bounds with the
fully-digital system. Plausible explanations for this observation include the use 
of small arrays at the user end ($2 \times 2$) and $L_k = 6$ clusters in the channel. 
A more complex hybrid precoding architecture achieved by
optimally choosing $\bF_{\sf Dig}$ with respect to some performance metric may assist
in bridging this gap. It is also to be pointed out
that while the alternate
optimization-based sum rate serves as an upper bound for most channel realizations,
for some realizations (especially at low ${\sf SINR}$ values where the ${\sf SLNR}$
optimization has a different behavior than the sum rate maximization), this connection
breaks down.

\begin{figure*}[htb!]
\begin{center}
\begin{tabular}{cc}
\includegraphics[height=2.4in,width=3.0in]{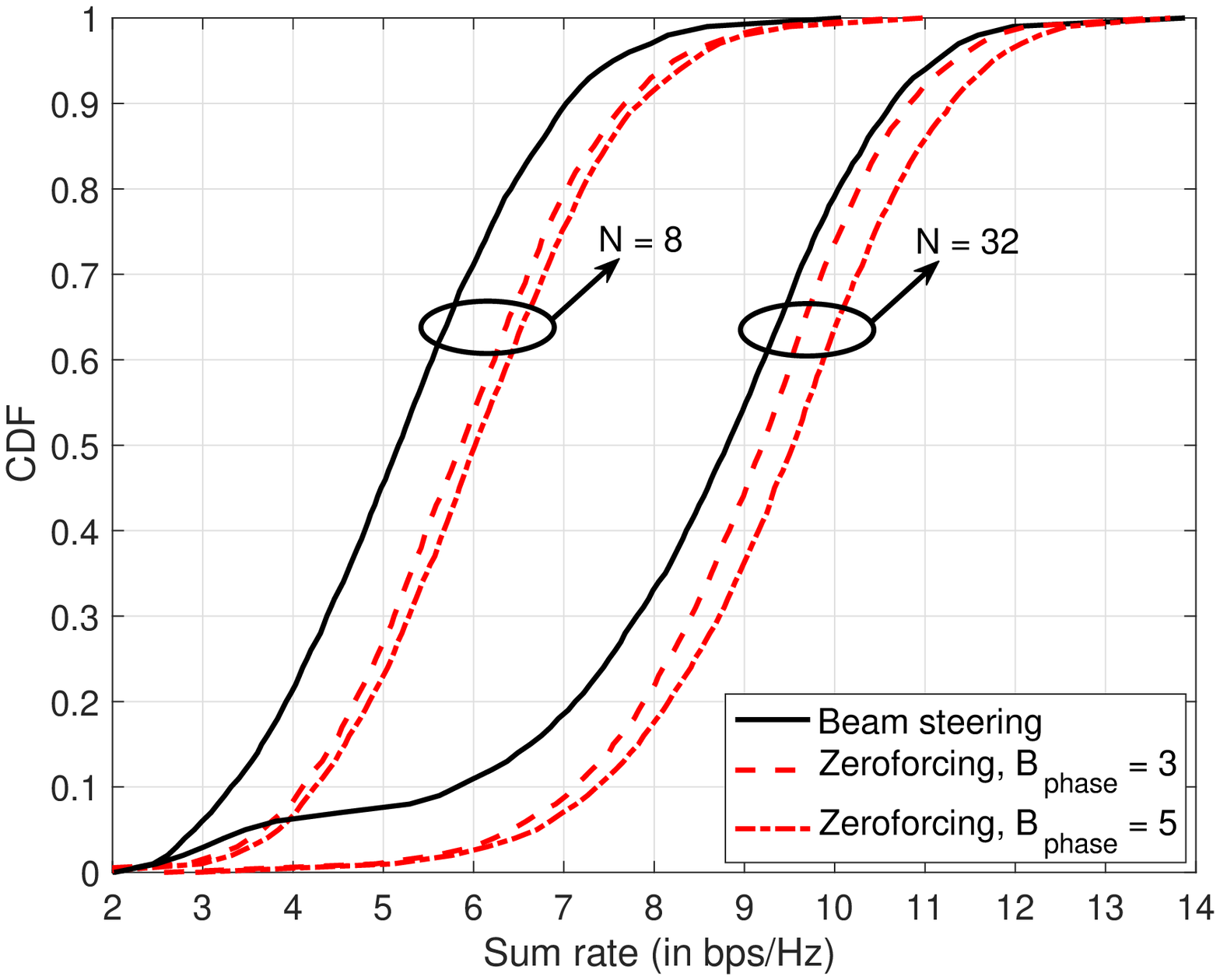}
& {\hspace{0.1in}}
\includegraphics[height=2.4in,width=3.0in]{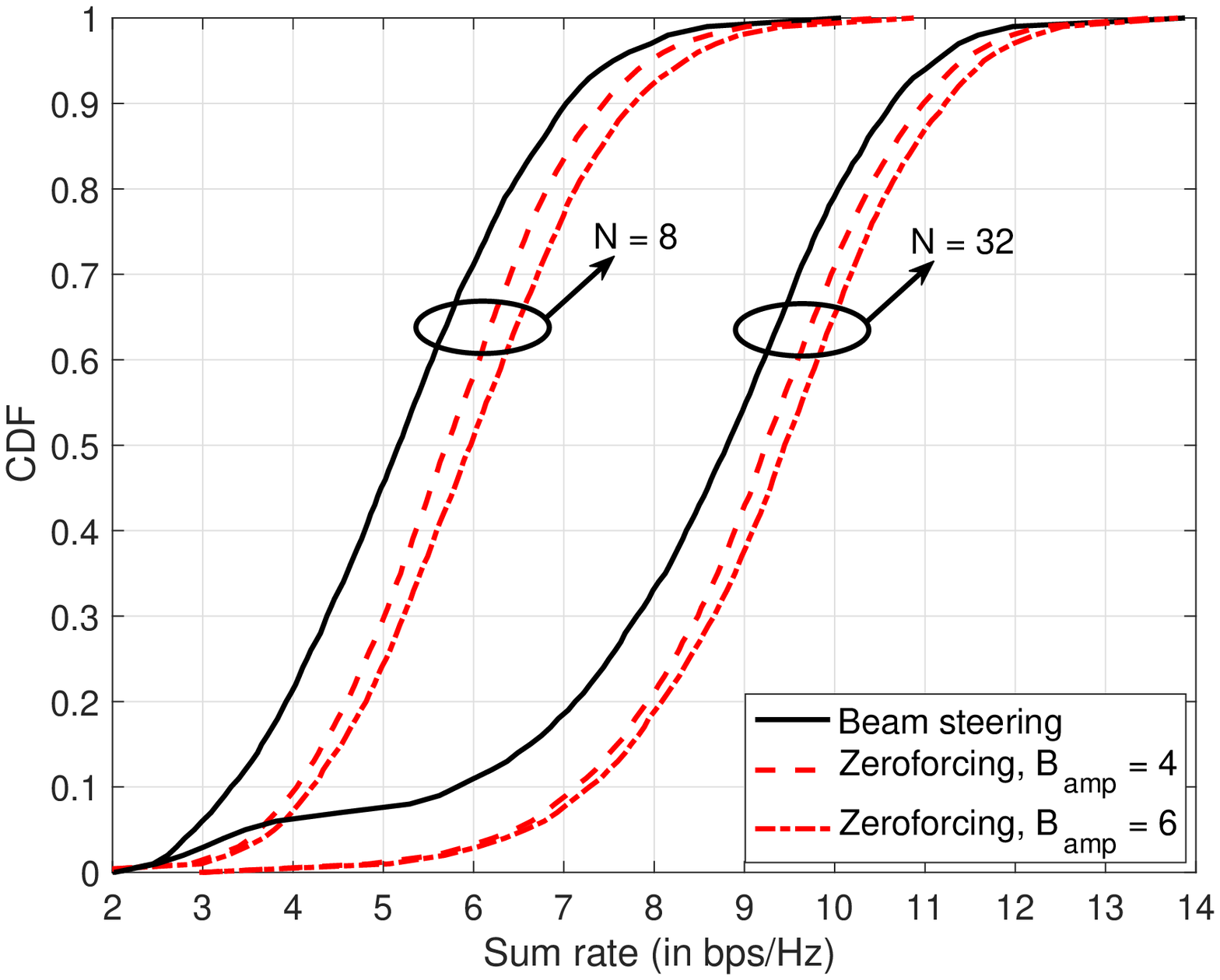}
\\
(a) & (b)
\end{tabular}
\caption{\label{fig:bfQuant}CDF of sum rates of the different schemes with
quantization constraints on the multi-user beam's (a) phases and (b) amplitudes.}
\end{center}
\end{figure*}

\begin{figure*}[htb!]
\begin{center}
\begin{tabular}{cc}
\includegraphics[height=2.4in,width=3.0in]{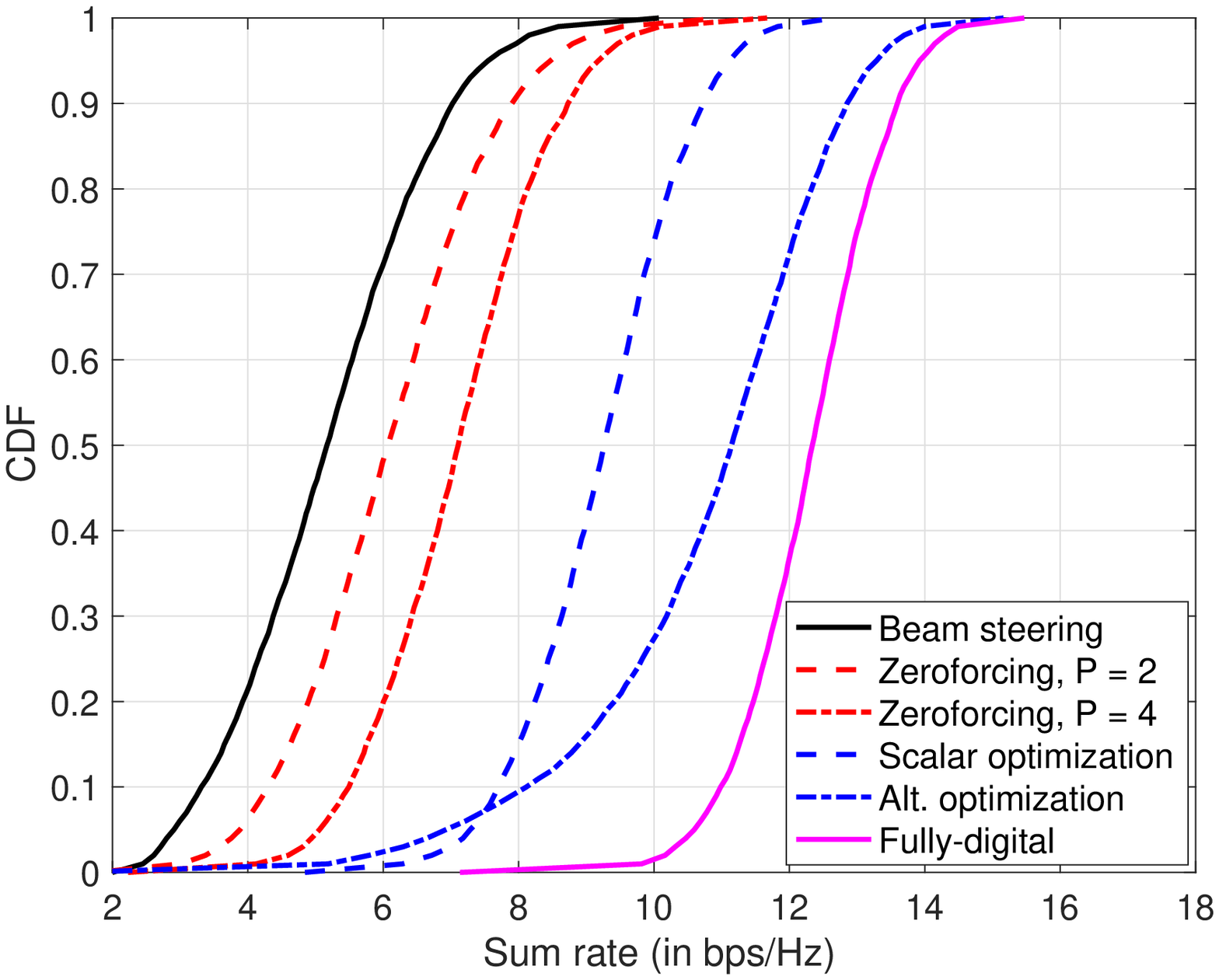}
& {\hspace{0.1in}}
\includegraphics[height=2.4in,width=3.0in]{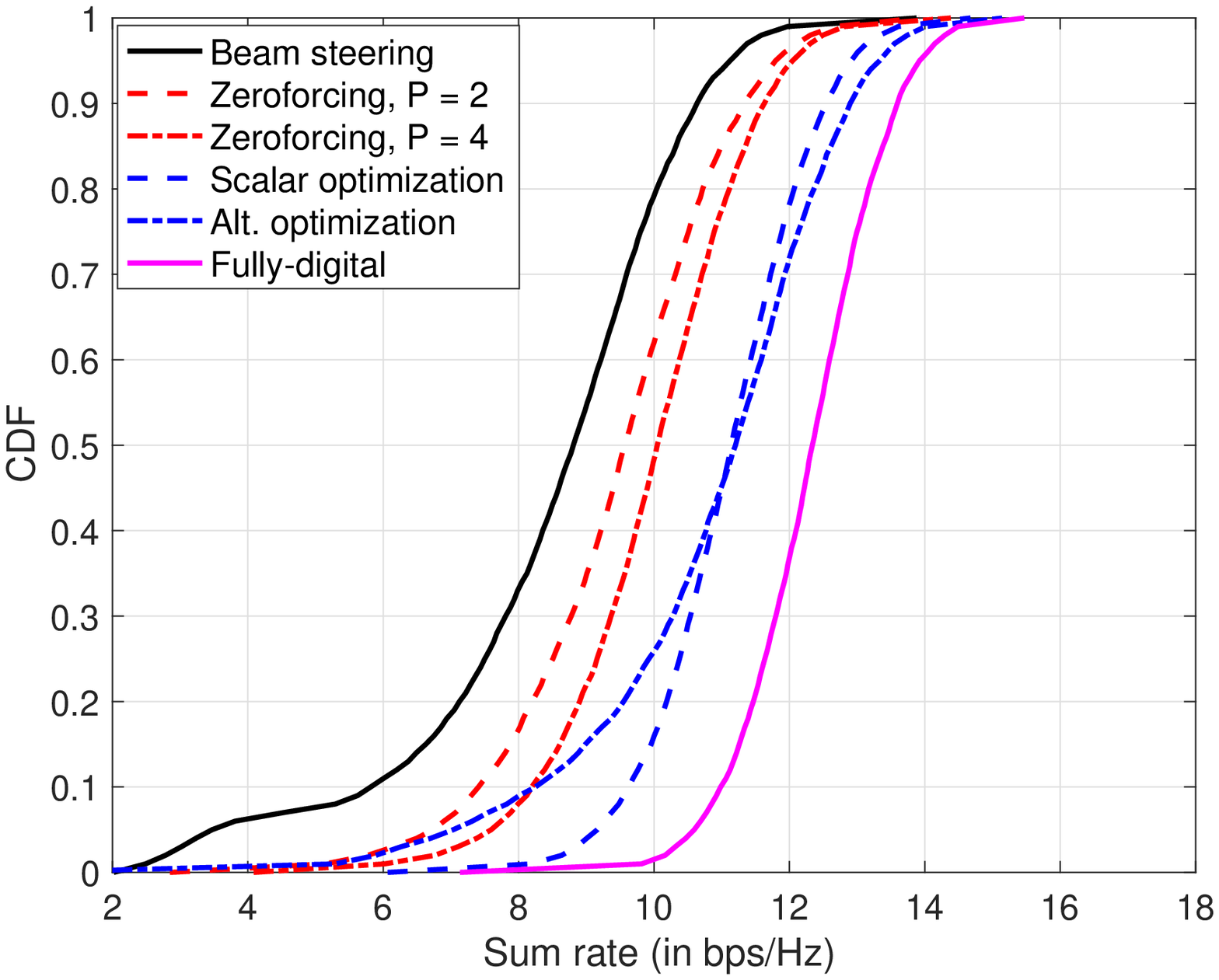}
\\
(a) & (b)  
%\\ 
%\includegraphics[height=2.4in,width=3.0in]{figures/sumrate_txCodebook_256_rxCodebook_16_upperBounds_v2.eps} %\\
%(c)
\end{tabular}
\caption{\label{fig:upperBounds}CDF of sum rates of the multi-user schemes
compared with the two upper bounds using a $M = 16$ codebook with (a) $N = 8$, (b) $N = 32$. 
%, (c) $N = 256$.
}
\end{center}
\end{figure*}

\begin{figure*}[htb!]
\begin{center}
\includegraphics[height=2.4in,width=3.0in]{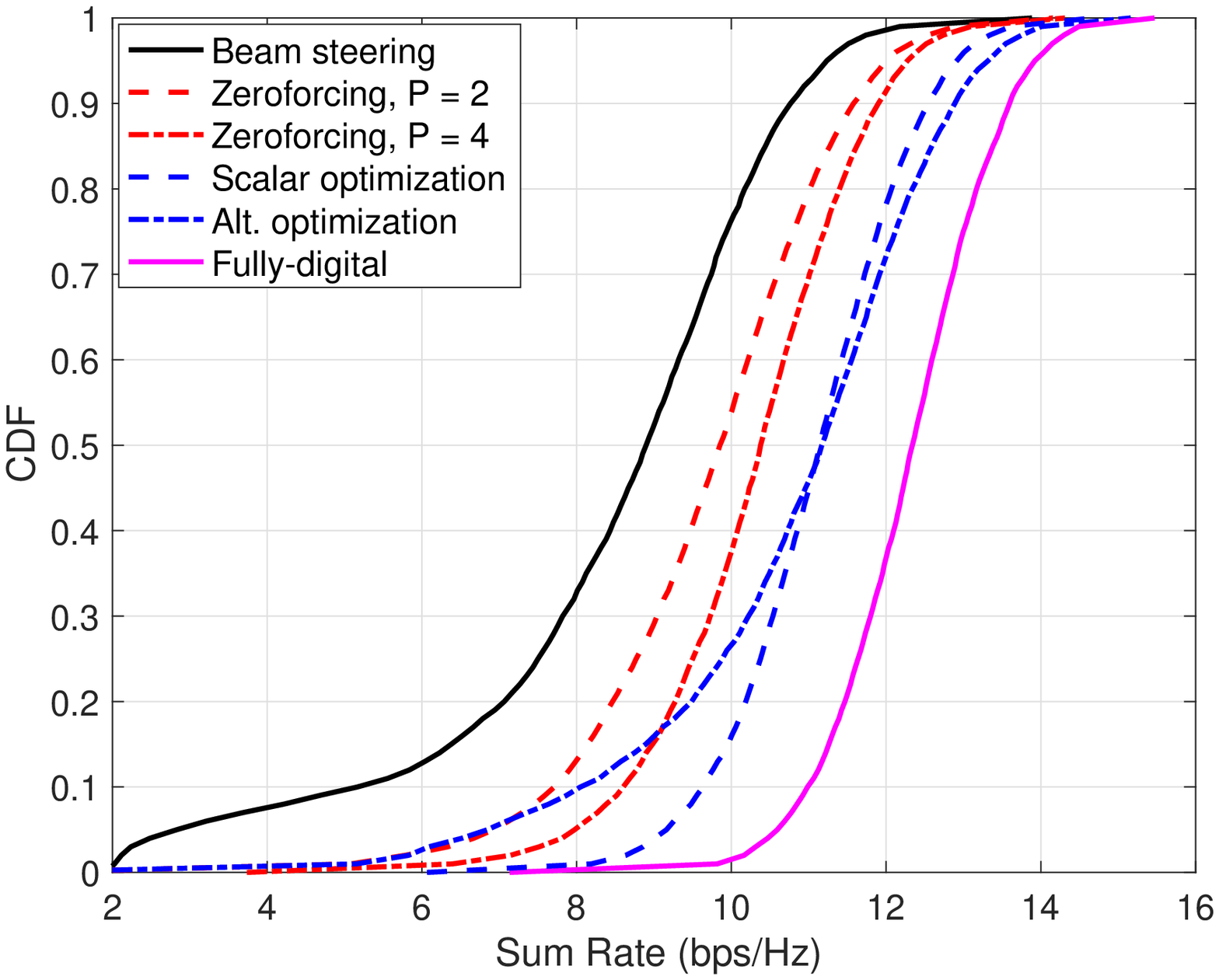} \caption{\label{fig:upperBounds1}
CDF of sum rates of the multi-user schemes
compared with the two upper bounds using a $M = 16$ codebook with $N = 256$.}
\end{center} 
\end{figure*} 

\ignore{
\begin{figure*}[htb!]
\begin{center}
\begin{tabular}{cc}
\includegraphics[height=2.5in,width=3.3in] {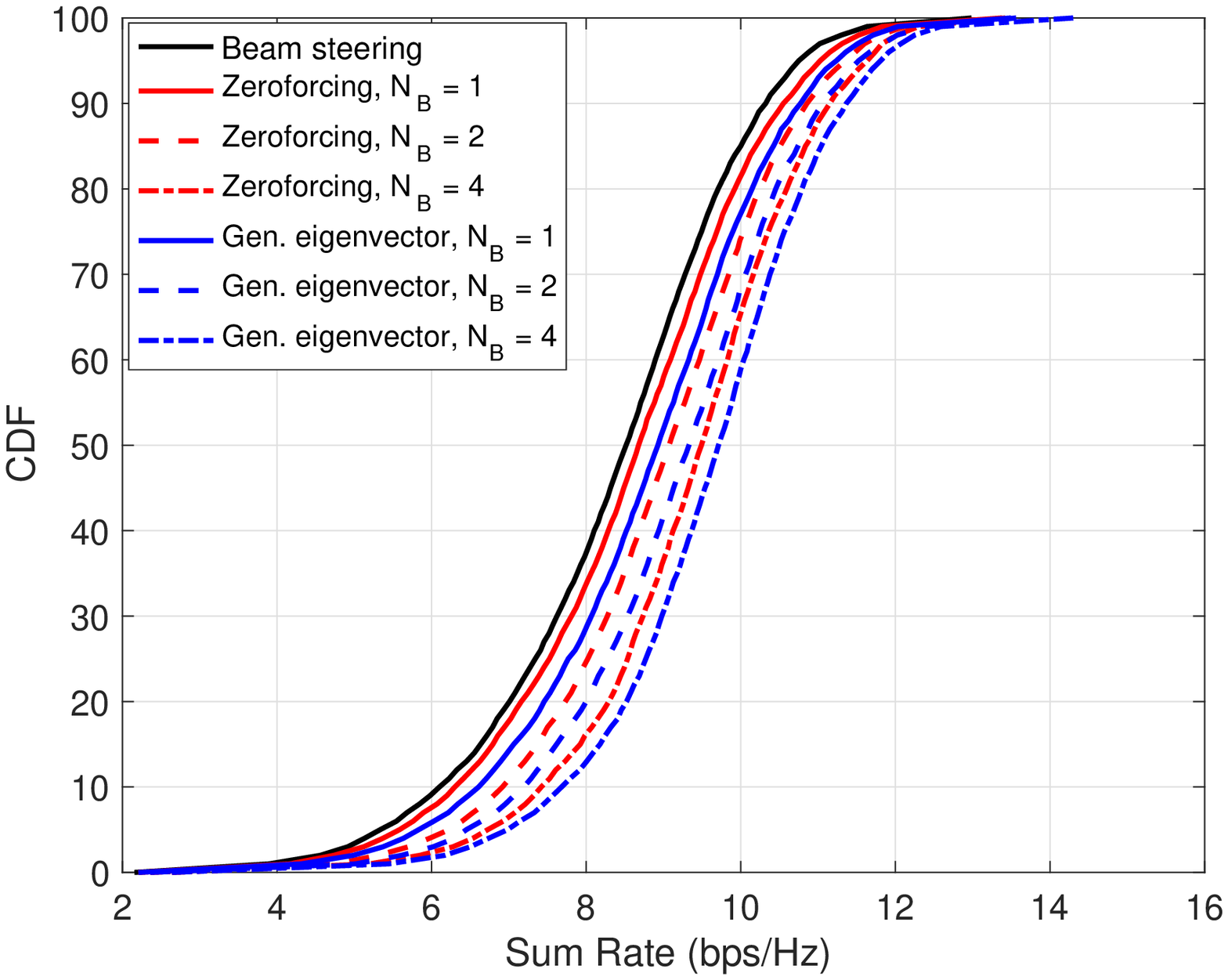}
&
%%%
\includegraphics[height=2.5in,width=3.3in] {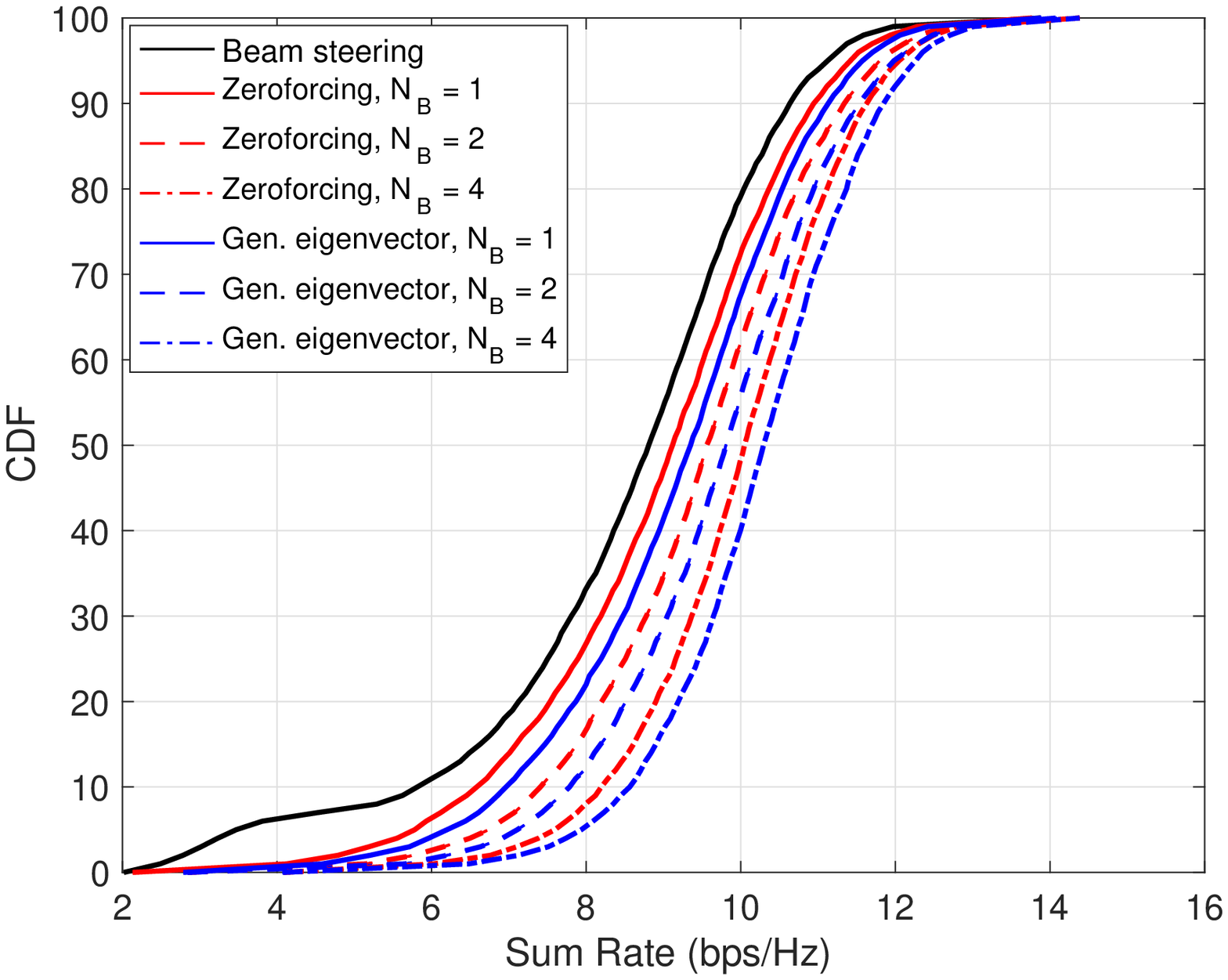}
\\
(a) & (b)
\end{tabular}
\caption{\label{fig:sumrates}
CDF of sum rates with MU-MIMO schemes using $M = 16$ for (a) $N = 16$ and
(b) $N = 32$.}
\end{center}
%\vspace{-5mm}
\end{figure*}
}

\section{Concluding Remarks}
\label{sec6}
The focus of this work has been the development of a feedback mechanism to convey estimates
of certain quantities of interest from an initial beam alignment phase to enable the
base-station to construct an advanced RF precoding structure for multi-user transmissions.
These quantities of interest include the top-$P$ (where $P \geq 1$) base-station side beam
indices, phases and amplitudes of an appropriate received signal estimate, as well as 
the cross-correlation
information of the beams at the user end. This feedback is leveraged to reconstruct/estimate
a rank-$P$ approximation of the channel matrix of interest at the base-station end and generate
a zeroforcing structure for multi-user interference management. Numerical studies show that the
additional feedback overhead is marginal, but the relative performance improvement over a simplistic
beam steering scheme is significant even with a very coarse initial beam alignment codebook.

This study reinforces the importance of the development of low-complexity (in terms of feedback
overhead as well as implementation) yet good (in terms of performance and structure) 
feedback techniques for
large-MIMO systems~\cite{david_review,vasanth_design}. While this work has only scratched the
surface of such techniques, a number of possible future research directions are worth considering.
Benchmarking the performance of any proposed feedback technique with a tight upper bound (for
the sum rate) is an area of fundamental difficulties due to the non-convex nature of the
problem~\cite{cioffi_weighted,kobayashi,vasanth_arxiv2011} and is richly rewarding. Understanding
the fundamental limits of hybrid precoders beyond the phase-only control architecture that is
common in the literature, as well as providing a directional intuition into the structure of the
precoder construction (in contrast to a {\em black box} optimization solution) are of importance
in practical implementations. While the solutions proposed in this work can be 
readily extended to polarization-diversity transmissions, extending it to the case where the
users possess two (or more) RF chains with the base-station communicating over two {\em spatial} 
layers is of
importance from a 3GPP 5G-NR deployment perspective. Study of different hybrid beamforming
architectures such as the sub-connected structure in~\cite{gao_hybrid} and comparison
with the proposed scheme(s) would be of interest. Sensitivity of such advanced schemes to
impairments such as Doppler and %phase noise~\cite{khanzadi}
phase noise are also worth exploring more carefully.

\appendix
\subsection{Generalized Eigenvector Solution}
\label{app_ge}
We need the following statement on the generalized eigenvector solution to the
standard optimization that will be repeatedly considered in this work.
\begin{lem}
\label{lem_ge_soln}
If ${\bf B}$ is an $n \times n$ positive definite matrix, then the principal
square-root (denoted as ${\bf B}^{1/2}$) exists and is invertible (denoted as
${\bf B}^{-1/2}$). Further, if ${\bf A}$ is another $n \times n$ positive
semi-definite matrix, the following optimization over $n \times 1$ unit-norm
vectors is well-understood~\cite{vasanth_arxiv2011,vasanth_jsac2017}
\begin{eqnarray}
{\bf f}_{\sf opt} =
\arg\max_{ {\bf f} \hsppp : \hsppp \| {\bf f}\| = 1}
\frac{ {\bf f}^{\dagger} {\bf A} {\bf f} }{ {\bf f}^{\dagger} {\bf B} {\bf f} }
= \frac{ {\bf B}^{-1/2} \cdot {\sf Dom \hspp eig} \left(
{\bf B}^{-1/2} \hsppp {\bf A} \hsppp {\bf B}^{-1/2} \right)
}
{ \| {\bf B}^{-1/2} \cdot {\sf Dom \hspp eig} \left(
{\bf B}^{-1/2} \hsppp {\bf A} \hsppp {\bf B}^{-1/2} \right)\| }
\label{eq_fopt}
\end{eqnarray}
with ${\sf Dom \hspp eig}(\cdot)$ denoting the dominant eigenvector operation of the
underlying matrix. In the special case where ${\bf A} = {\bf w} {\bf w}^{\dagger}$
is a rank-$1$ matrix for some column vector ${\bf w}$, then ${\bf f}_{\sf opt}$ reduces
to ${\bf f}_{\sf opt} = \frac{ {\bf B}^{-1} {\bf w}} { \| {\bf B}^{-1} {\bf w}\|}$.
\qed
\end{lem}
Note that the generalized eigenvector of a matrix pair $({\bf A}, {\bf B})$ is a vector
${\bf x}$ that solves the problem ${\bf A} {\bf x} = \sigma {\bf B}{\bf x}$ for
some scalar $\sigma$. From this description, it can be seen that ${\bf f}_{\sf opt}$
in~(\ref{eq_fopt}) is the dominant unit-norm generalized eigenvector of the matrix
pair $({\bf A}, {\bf B})$.

\subsection{Proof of Prop.~\ref{prop_zf}}
\label{proof_prop_zf}
Given the expression for $\widehat{\sf SINR}_m$ in~(\ref{eq_estimated_SINR}), the
zeroforcing structure corresponds to the construction $\{ {\bf f}_m \}$ such that
\begin{eqnarray}
\big|{\bf g}_k^{\dagger} \widehat{\bf H}_k \hsppp {\bf f}_m \big|^2 = 0, \hspp m \neq k,
\hspp \{m , k\} \in 1, \cdots, K.
\label{eq_zf_condition}
\end{eqnarray}
An elementary computation shows that by setting ${\bf f}_m, \hsppp m = 1, \cdots, K$ as in
the statement of the proposition, we can ensure the condition in~(\ref{eq_zf_condition}).
\qed

\subsection{Proof of Prop.~\ref{prop_ge}}
\label{proof_prop_ge}
%\noindent {\bf \em \underline{Main Proof:}}
Since ${\bf f}_k^{\dagger} {\bf f}_k = 1$, we can write $\widehat{ {\sf SLNR} }_k$ as
\begin{eqnarray}
\widehat{ {\sf SLNR} }_k = \frac{ \eta_{k,k} \cdot {\bf f}_k^{\dagger}
\cdot \left( \widehat{\bf H}_k^{\dagger} {\bf g}_k {\bf g}_k^{\dagger} \widehat{\bf H}_k
\right) \cdot {\bf f}_k  }
{ {\bf f}_k^{\dagger} \cdot \left( {\bf I}_{ N_{\sf t}}
+ \sum_{m \neq k} \eta_{m,k} \hsppp
\widehat{\bf H}_m^{\dagger} {\bf g}_m {\bf g}_m^{\dagger} \widehat{\bf H}_m
\right) \cdot {\bf f}_k } .
\end{eqnarray}
The optimal structure of ${\bf f}_k$ in the statement of the proposition follows
directly from Lemma~\ref{lem_ge_soln}. %these facts.
\qed

%{\vspace{-0.05in}}
\bibliographystyle{IEEEbib}
\bibliography{newrefsx2}

\end{document}